%% file: main.tex
\documentclass{article}

\usepackage[preprint, nonatbib]{neurips_2022}

\usepackage[utf8]{inputenc} 
\usepackage[T1]{fontenc}    
\usepackage[hidelinks]{hyperref}       
\usepackage{url}            
\usepackage{booktabs}       
\usepackage{amsfonts}       
\usepackage{nicefrac}       
\usepackage{microtype}      
\usepackage[dvipsnames]{xcolor}         
\usepackage{amsmath,amsthm,amsfonts,amssymb}

\usepackage{wrapfig}

\usepackage{amsmath}
\usepackage{amssymb}
\usepackage{bbm}

\usepackage[backend=biber, style=ieee]{biblatex}
\usepackage{amsthm}

\theoremstyle{definition}
\newtheorem{definition}{Definition}
\newtheorem{problem}{Problem}

\theoremstyle{plain}
\newtheorem{theorem}{Theorem}
\newtheorem{corollary}[theorem]{Corollary}
\newtheorem{proposition}[theorem]{Proposition}
\newtheorem{lemma}[theorem]{Lemma}

\theoremstyle{remark}
\newtheorem{remark}{Remark}

\usepackage{glossaries}
\usepackage{todonotes}
\usepackage{url}
\usepackage{siunitx}

\usepackage{algorithm}
\usepackage{algpseudocode}
\algnewcommand\algorithmicleftcomment[1]{\(\triangleright\) \textit{#1}}%

\algrenewcommand\algorithmiccomment[1]{\hfill\(\triangleright\) \textit{#1}}%
\algdef{S}[WHILE]{WhileNoDo}[1]{\algorithmicwhile\ #1}%
   
\usepackage{multicol}

\usepackage{caption}
\usepackage{subcaption}

\DeclareMathOperator*{\argmax}{arg\,max}

\DeclareMathOperator*{\softmax}{softmax}

\newcommand{\LL}[1]{{\color{blue}[LL: #1]}}

\input{acronyms.tex}

\title{Safety Certification for Stochastic Systems via Neural Barrier Functions}
\author{%
  Frederik Baymler Mathiesen \\
  Delft Center for Systems \& Control\\
  Delft University of Technology\\
  \texttt{f.b.mathiesen@tudelft.nl} \\
  \And
   Simeon Calvert \\
  Department for Transport \& Planning\\
  Delft University of Technology\\
  \texttt{s.c.calvert@tudelft.nl} \\
   \And
   Luca Laurenti \\
  Delft Center for Systems \& Control\\
  Delft University of Technology\\
  \texttt{l.laurenti@tudelft.nl} \\
}

\begin{document}

\maketitle

\begin{abstract}

Providing non-trivial certificates of safety for non-linear stochastic systems is an important open problem that limits the wider adoption of autonomous systems in safety-critical applications. One promising solution to address this problem is barrier functions. The composition of a barrier function with a stochastic system forms a supermartingale, thus enabling the computation of the probability that the system stays in a safe set over a finite time horizon via martingale inequalities. However, existing approaches to find barrier functions for stochastic systems generally rely on convex optimization programs that restrict the search of a barrier to a small class of functions such as low degree SoS polynomials and can be computationally expensive. In this paper, we parameterize a barrier function as a neural network and show that techniques for robust training of neural networks can be successfully employed to find neural barrier functions. Specifically, we leverage bound propagation techniques to certify that a neural network satisfies the conditions to be a barrier function via linear programming and then employ the resulting bounds at training time to enforce the satisfaction of these conditions. We also present a branch-and-bound scheme that makes the certification framework scalable. We show that our approach outperforms existing methods in several case studies and often returns certificates of safety that are orders of magnitude larger.
\end{abstract}

\section{Introduction}
Modern autonomous systems are inherently non-linear, incorporate feedback controllers often trained with machine learning techniques, and are subject to uncertainty due to unmodelled dynamics or exogenous disturbances \cite{duriez2017machine}. Despite these complexities, autonomous systems are often employed in safety-critical applications, such as autonomous cars \cite{gordon2015} or air traffic control \cite{wise1991}, where a failure of the system can have catastrophic consequences \cite{Wilko2018, wise1991}. For these applications, computation of probabilistic safety, defined as the probability that the system cannot evolve over time to an unsafe region of the state space, is of paramount importance. Despite the recent efforts \cite{laurenti2020formal, belta2019}, computing certificates that guarantee that probabilistic safety is above a certain threshold  still remains a particularly challenging problem, mainly due to the complexity of these system.

A promising approach for safety verification is the employment of \emph{barrier functions} \cite{ames2019}. Similar to the Lyapunov function approach for proving stability \cite{TanSearchingFC}, barrier functions aims to prove temporal properties of a system without the need to explicitly analyze the flow of the system \cite{ames2019}. In the stochastic setting, a barrier function is a function that when composed with the system forms a non-negative \emph{supermartingale} or a \emph{$c$-martingale} \cite{prajna2007framework}. Then, martingale inequalities can be employed to compute (a lower bound on) the probability that the system remains safe over time \cite{kushner1967stochastic}. The main challenge with this approach is to find a barrier function for a specific system maximizing the certified lower bound of safety. Existing approaches generally formulate the search of a barrier function for a stochastic system as a convex optimization problem by restricting the search of a barrier over a limited class of functions, generally exponential \cite{steinhardt2012finite} or relatively low-degree SOS polynomials \cite{SANTOYO2021109439}, which often leads to overly conservative safety estimates. In this context, neural networks hold great potential due to their universal approximation power and their training flexibility. However, while neural network barrier functions, or simply \gls{nbf}, have already been considered for deterministic systems \cite{dawson2022, Abate2021, Zhao2020}, still no work has explored the potential of neural networks to provide safety certificates for stochastic systems.

\begin{figure}
    \centering
    \begin{subfigure}[b]{0.33\textwidth}
        \centering
        \includegraphics[width=\textwidth]{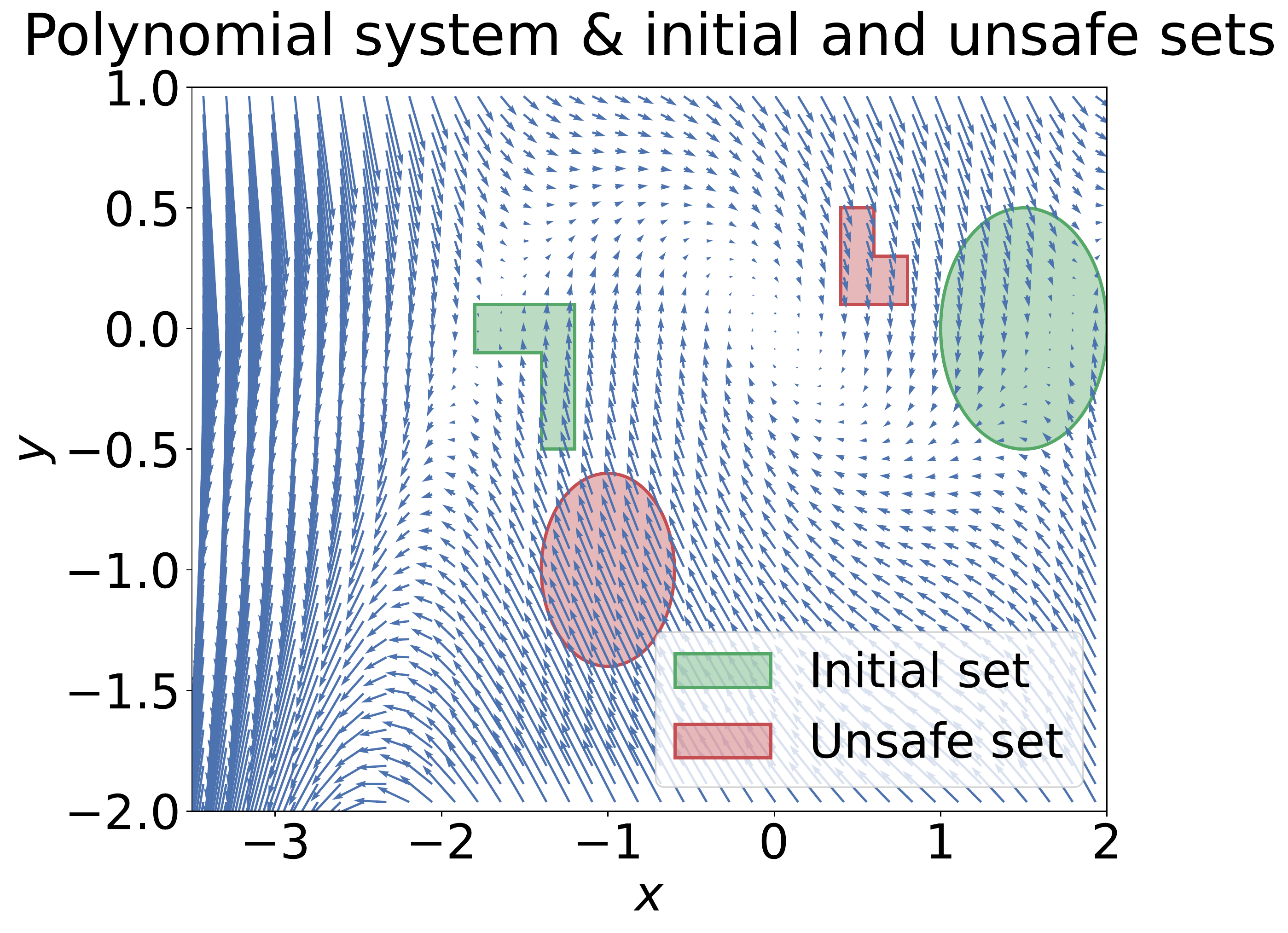}
        \caption{}
    \end{subfigure}\hfill
    \begin{subfigure}[b]{0.33\textwidth}
        \centering
        \includegraphics[width=0.9\textwidth]{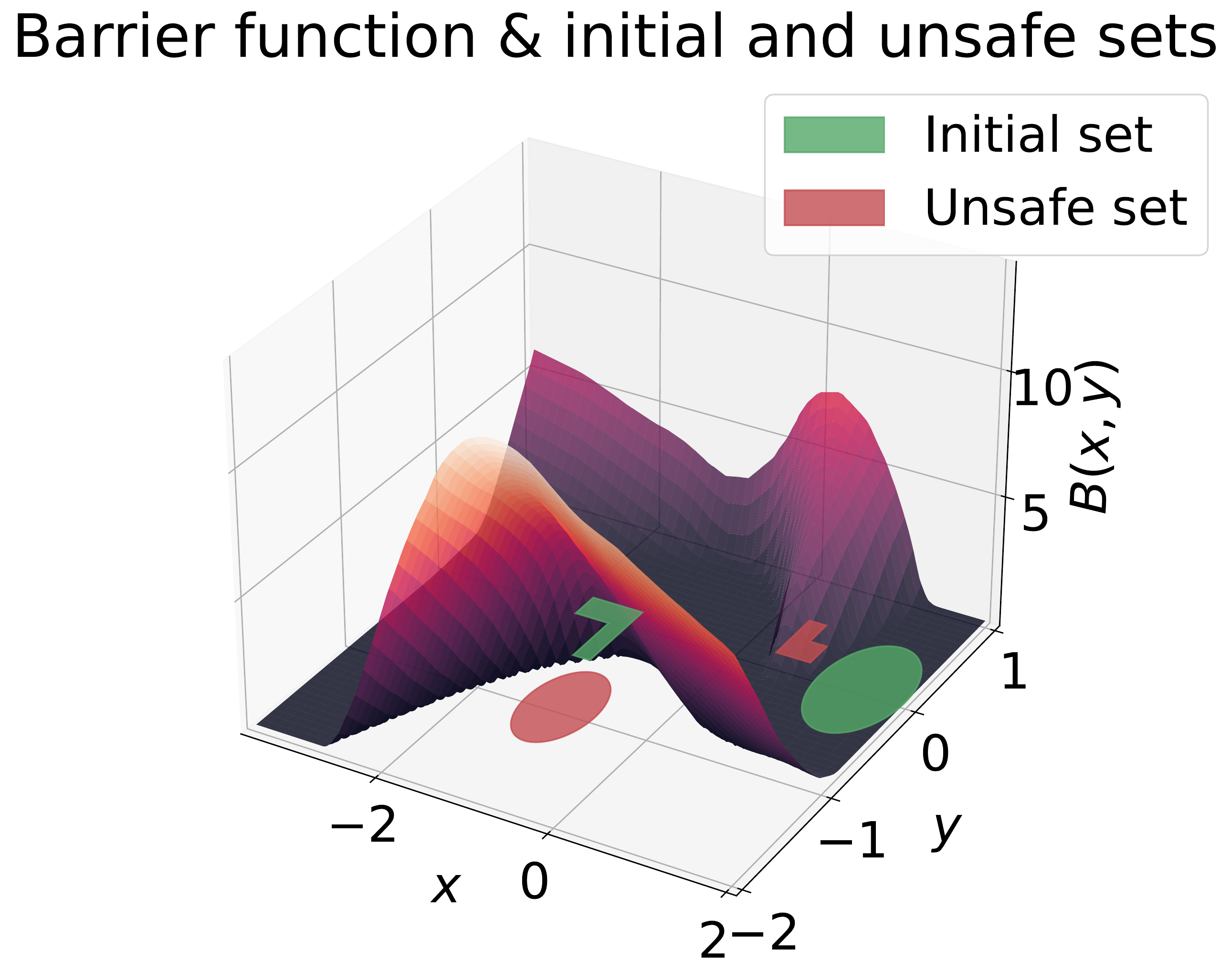}
        \caption{}\label{fig:polynomial_barrier}
    \end{subfigure}\hfill
    \begin{subfigure}[b]{0.33\textwidth}
        \centering
        \includegraphics[width=0.8\textwidth]{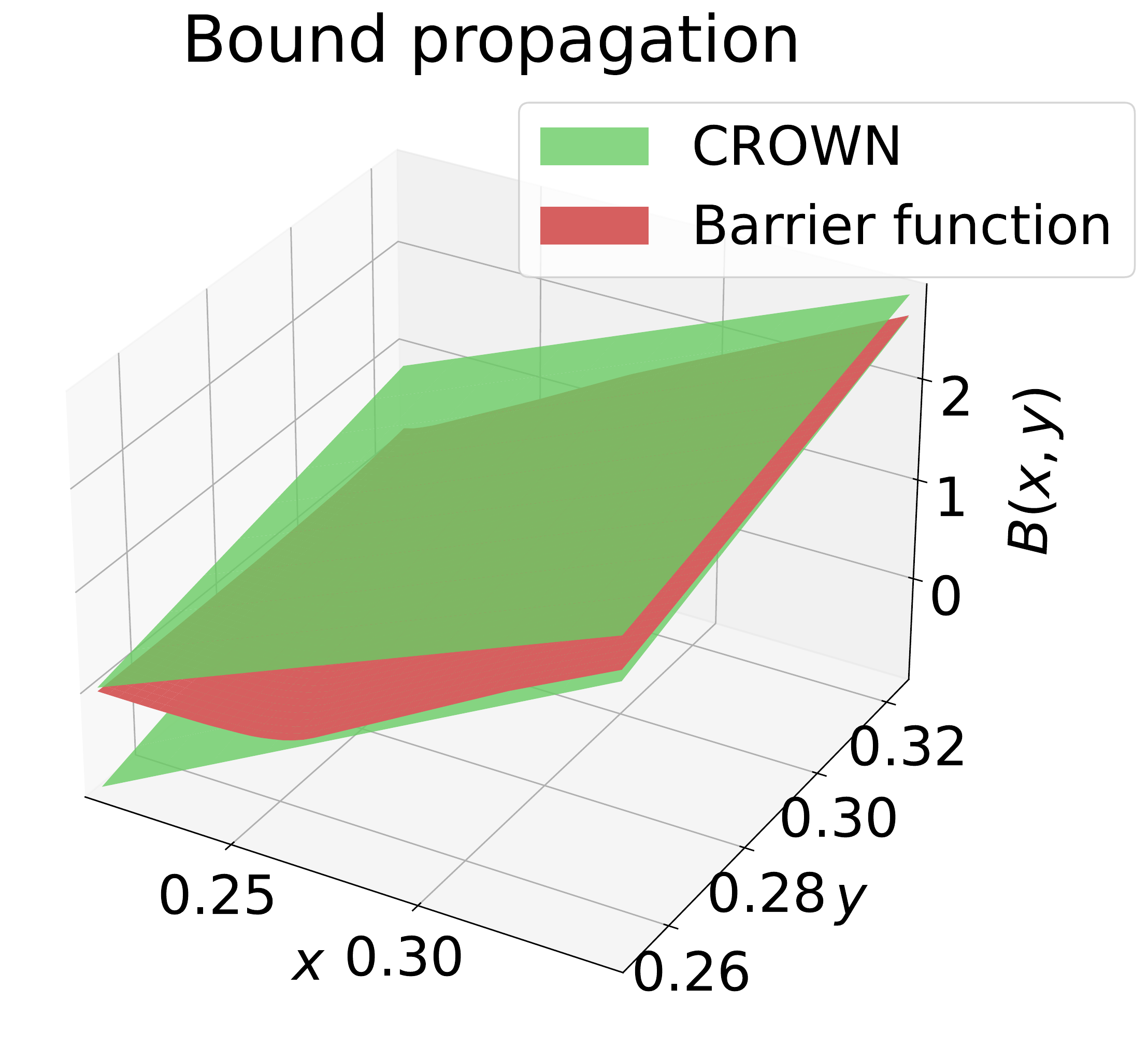}
        \caption{}
    \end{subfigure}\hfill
    \begin{subfigure}[b]{0.33\textwidth}
        \centering
        \includegraphics[width=\textwidth]{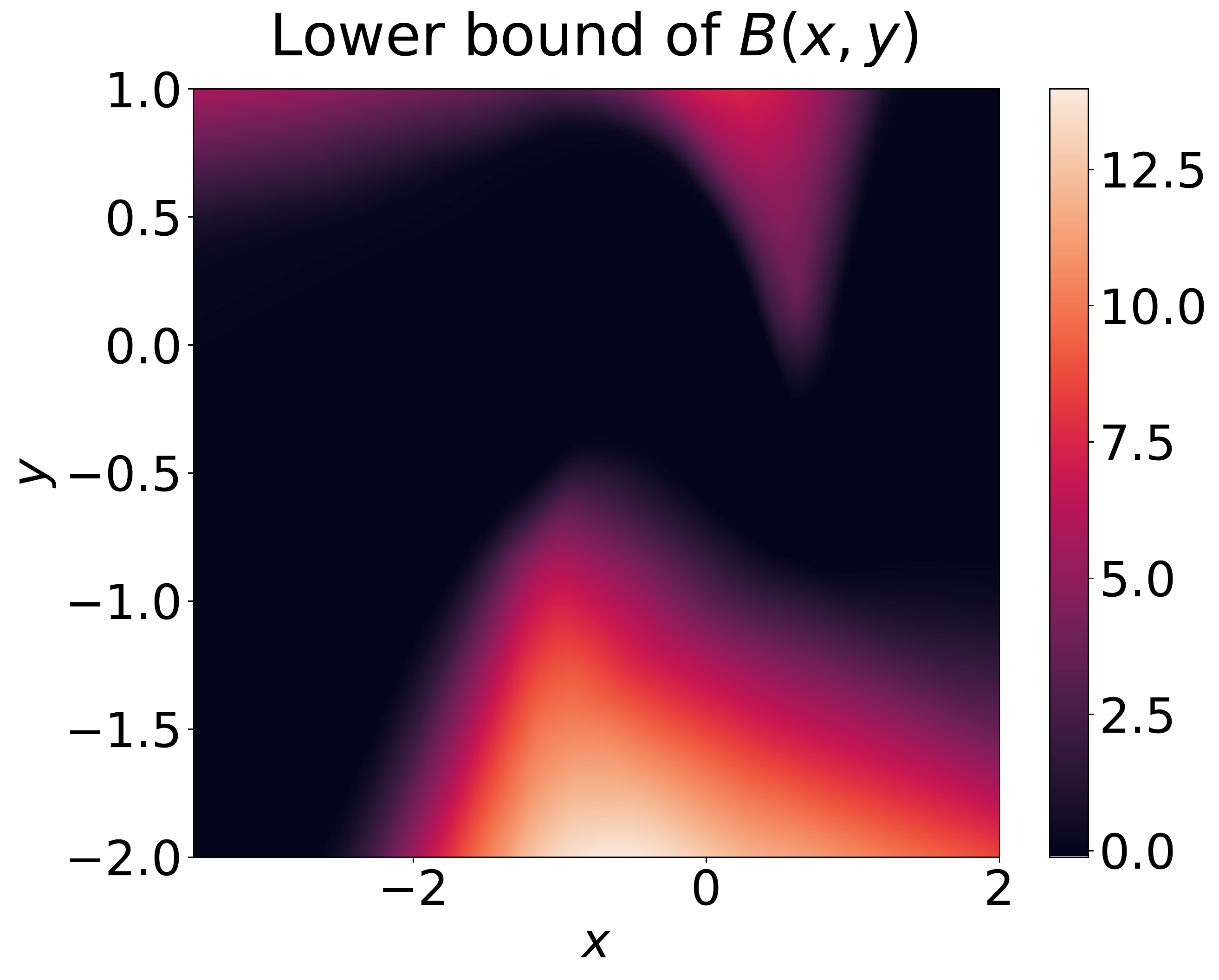}
        \caption{}
    \end{subfigure}\hfill
    \begin{subfigure}[b]{0.33\textwidth}
        \centering
        \includegraphics[width=\textwidth]{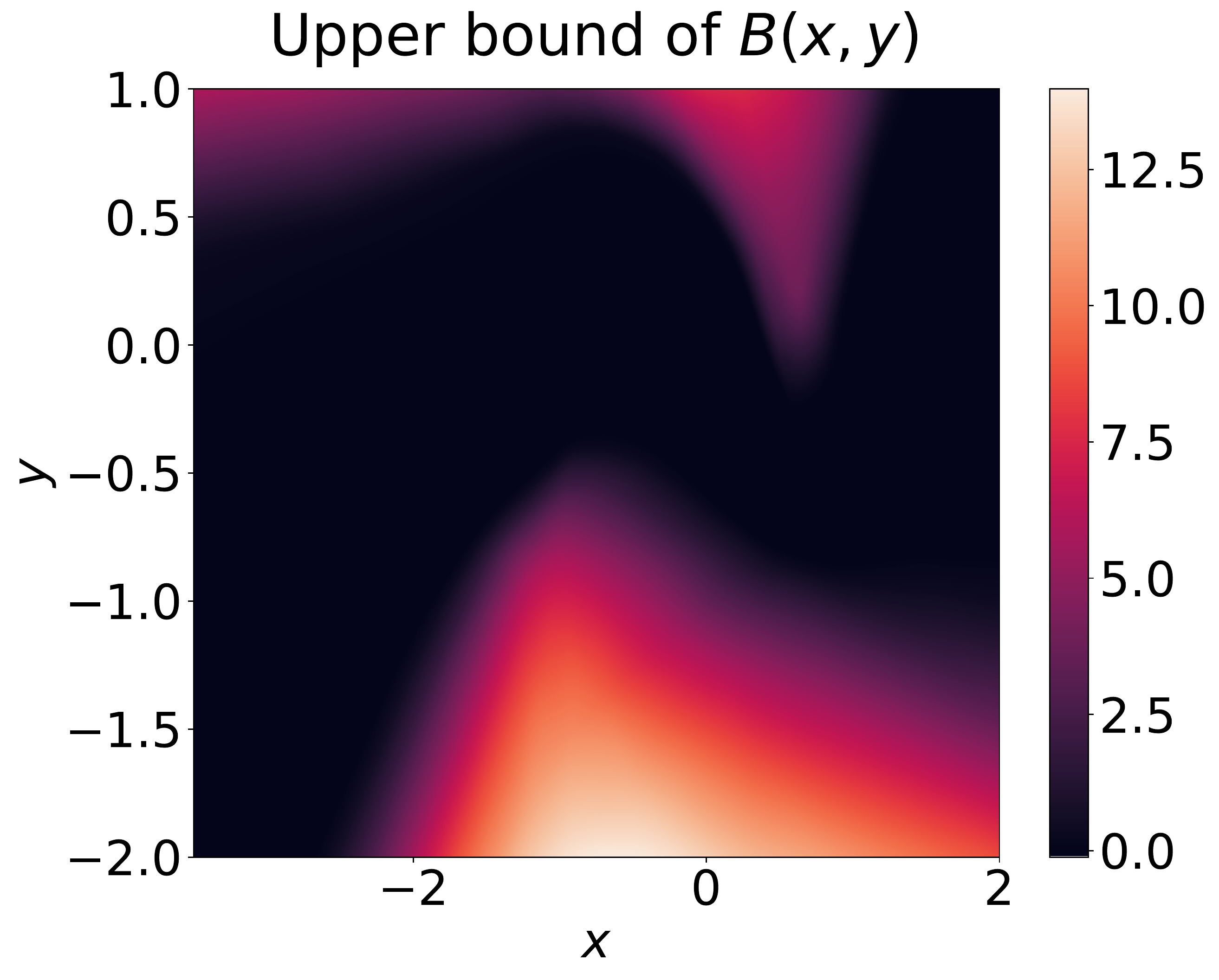}
        \caption{}
    \end{subfigure}\hfill
    \begin{subfigure}[b]{0.33\textwidth}
        \centering
        \includegraphics[width=0.95\textwidth]{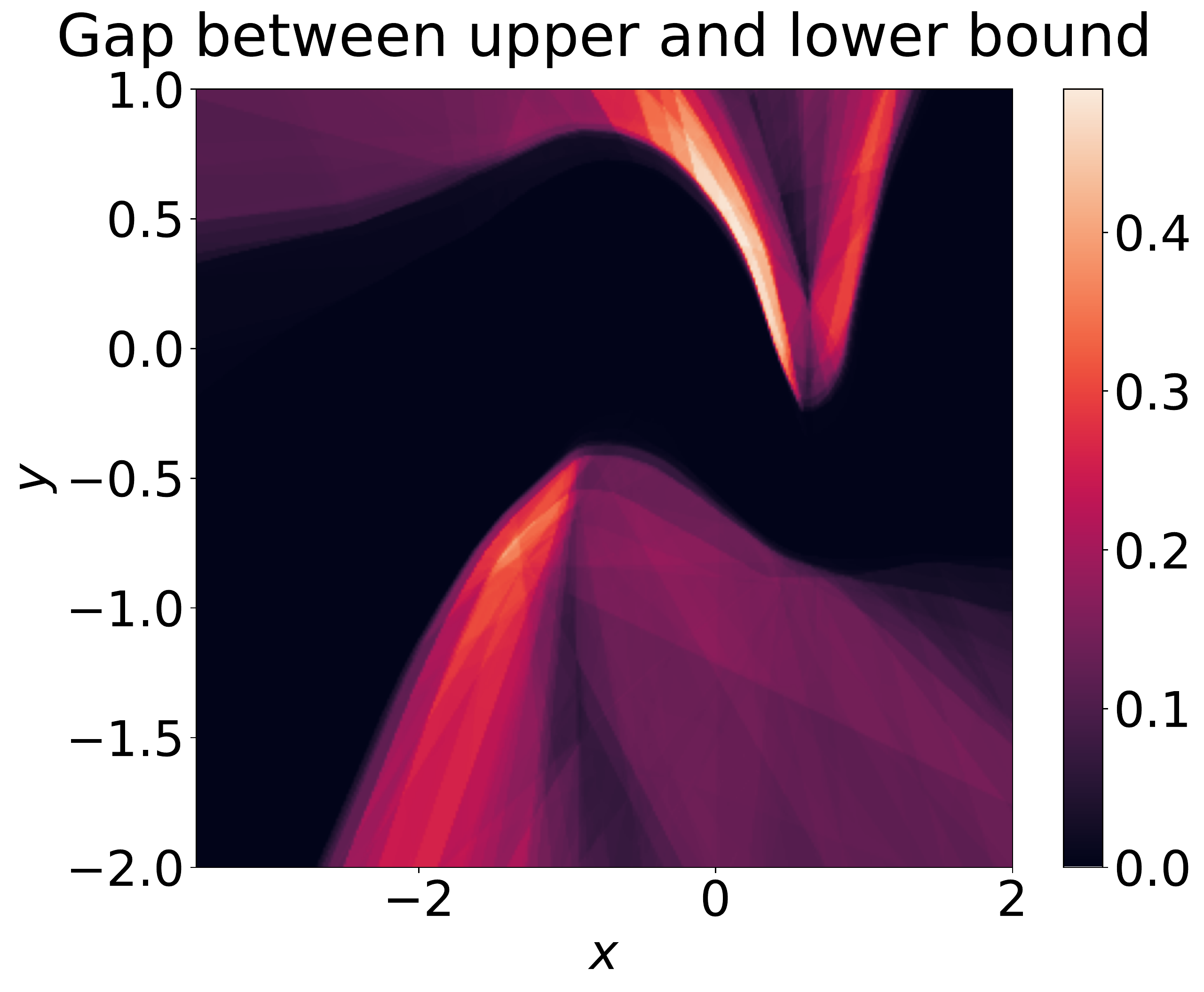}
        \caption{}
    \end{subfigure}
    \caption{Polynomial system with additive Gaussian noise adapted from \cite{Abate2021} and corresponding Neural Barrier Function analyzed using bound propagation. \textbf{(a)} Plot of the nominal system, i.e. of the vector field without noise, and of initial set and unsafe region. 
    \textbf{(b)} Neural Barrier Function (NBF) as synthesized by our framework relative to initial and unsafe sets. The composition of the NBF with the system forms a $c-$martingale.
    \textbf{(c)} Linear relaxation of the Neural Barrier Function for a single hyperrectangular region using CROWN \cite{zhang2018efficient}. 
    \textbf{(d, e)} Bounds of our Neural Barrier Function computed partitioning the state space in a $320 \times 320$ grid and using linear bound propagation in each partition. \textbf{(f)} The gap between the upper and lower bounds. We will use this gap to drive our partitioning.}
    \label{fig:polynomial_system_barrier}
\end{figure}

In this paper, we consider non-linear stochastic systems with additive noise and propose an algorithmic framework to train \glspl{nbf} for these systems. In order to do so, by relying on recent techniques for linear relaxation of neural networks \cite{zhang2018efficient} and the linearity of the expectation operator, we show that the problem of certifying that a neural network is a barrier function, and computing the resulting safety probability, can be relaxed to the solution of a set of linear programs. Specifically, we partition the state space and leverage linear bound propagation techniques \cite{zhang2018efficient, xu2020} to find local linear lower and upper bounds of a neural network. 
For each region in the partition, we use these bounds prove that the composition of the neural network and the system forms a $c$-martingale. 
A branch-and-bound approach is then proposed to improve the scalability of our framework and adaptively refine the partition of the state space. At training time, as common in adversarial training of neural networks \cite{wicker2021bayesian}, we rely on the resulting linear bounds obtained by our certification framework to train a neural network that satisfies the $c$-martingale conditions, while also minimizing the conservatives of the resulting safety certification. An example of our framework is reported in Figure \ref{fig:polynomial_system_barrier}.

We experimentally investigate the suitability of our framework on several challenging non-linear models including a non-linear vehicle dynamics model. We find that our framework consistently outperforms existing state-of-the-art approaches based on \gls{sos} optimization \cite{SANTOYO2021109439}. For instance the experiments on the vehicle dynamics model show that, while the \gls{sos} approach returns a lower bound of safety of $0\%$ or fails due to computational constraints, our method is able to returns a certificate of safety of $87\%$ for the system by employing a neural network barrier function of 3 hidden layers and 128 neurons per layer.

In summary, this paper makes the following main contributions: (i)  We introduce a novel framework to train \glspl{nbf} for a given non-linear stochastic system, (ii) we present a branch-and-bound scheme for the computation of a lower bound on the safety probability via \glspl{nbf} based on the solution of linear programs, and (iii) on multiple benchmarks we show that our framework can train and certify \glspl{nbf} with multiple hidden layers of hundreds of neurons and substantially outperforms state-of-the-art competitive methods.

\subsection{Related works}

\paragraph{Safety certification of dynamical systems}
Safety guarantees for dynamical systems can be generally obtained with two different approaches: abstraction-based methods where the system is abstracted into a transition system, and verified using model checking \cite{laurenti2020formal, Cauchi2019, Prabhakar2013}, and barrier function-based methods where certification proceeds by finding an energy-like function to avoid the need for finding an analytical solution to differential or difference equation \cite{ames2019, Agrawal2017DiscreteCB}.
In the stochastic setting sum-of-squares optimization is the state-of-the-art method for finding polynomial stochastic barrier functions, and they have been already applied to discrete time \cite{steinhardt2012finite, SANTOYO2021109439}, continuous time \cite{SANTOYO2021109439}, hybrid systems \cite{prajna2007framework}, and also for LTL specifications \cite{9157966}.
    
\paragraph{Neural certificates}
Neural networks for representing barrier functions or Lyapunov function are collectively called neural certificates \cite{dawson2022}.
Neural Lyapunov functions were first proposed in \cite{1555943}, but has later been rediscovered and seen a surge \cite{Richards2018TheLN, Abate2021, Zhao2020}.
Lyapunov functions are designed to certify stability \cite{Richards2018TheLN, dawson2022, Abate2021, Jin2020NeuralCF}, while barrier functions are intended to certify safety \cite{dawson2022, Abate2021, Zhao2020}.
The condition required to check that a given neural network is a barrier function differs according to the class of systems considered; for deterministic systems \cite{chang2019neural, Abate2021, Zhao2020}, it is common to use zeroing barriers while supermartingale-based barriers are common for stochastic systems \cite{lechner2021}.
The majority of literature on neural certificates are for deterministic systems, where both continuous-time \cite{chang2019neural, Abate2021, Jin2020NeuralCF, Robey2020LearningCB, dawson2021safe, Zhao2020}, and discrete-time systems  have been studied \cite{Richards2018TheLN, lechner2021}.
To the best of our knowledge, the only paper that focuses on neural certificates for stochastic systems is \cite{lechner2021}. However,  \cite{lechner2021} only considers almost sure asymptotic stability, which is a different, and arguably simpler, problem than the one considered in this paper. Furthermore, in order to certify that the composition of a neural network with a non-linear system gives a valid Lyapunov function, \cite{lechner2021} relies on the Lipschitz constant of the underlying system. In contrast, our approach based on branch-and-bound and linear relaxations is different and may lead to less conservative bounds,
as already observed in verification of NNs when comparing these different certification approaches \cite{zhang2018efficient}.

A major challenge for neural certificates is the verification that the neural network candidate is indeed a valid barrier or Lyapunov function. 
For this reason a plethora of certification methods have been developed, including \gls{smt}-based certification for feed-forward neural networks with general action functions \cite{Abate2021, Zhao2020, chang2019neural, dawson2022}, \gls{milp}-based certification for piecewise affine functions \cite{dai2021lyapunov, 9304201, dawson2022}, finding Lipschitz constants over a grid mesh for the state space \cite{Richards2018TheLN, Jin2020NeuralCF, lechner2021}, and lastly sampling-based methods for validation of certificate conditions but not certification \cite{dawson2021safe}.
\gls{smt}- and \gls{milp}-based certification suffers from lack of scalability, restricting network sizes to 20-30 neurons in 2-3 hidden layers \cite{Abate2021, dai2021lyapunov}.
On the other hand, the Lipschitz method is computationally scalable to large neural networks, but generally very conservative \cite{Fazlyab2019}.
We address scalability with bound propagation, while ameliorating conservativity with a branch-and-bound partitioning scheme. As shown in Section \ref{sec:experiments} our resulting framework can scale to neural networks with multiple hidden layers and hundreds of neurons per layer.

\section{Problem formulation}
\label{sec:ProbForm}

We consider a stochastic discrete-time system described by the following stochastic difference equation:
\begin{equation}
\label{eq:system_equation}
    \mathbf{x}[k+1] = F(\mathbf{x}[k])+ \mathbf{v}[k] 
\end{equation}
where $\mathbf{x}[k]\in \mathbb{R}^n$ is the state of the system at time $k$, 
and $\mathbf{v}[k] $ is an independent random variable distributed according to $p(v)$ over an uncertainty space $V \subseteq \mathbb{R}^n$. 
The function $F : \mathbb{R}^n \times U \to X$ is a continuous function representing the one-step dynamics of System  \eqref{eq:system_equation}.  System  \eqref{eq:system_equation} represents a general model of non-linear stochastic system with additive noise, a class of stochastic systems widely used in many areas \cite{girard2002gaussian,jackson2021strategy}, which also includes non-linear systems in closed loop with feedback controllers synthesized with standard control theory methods (e.g. LQR \cite{aastrom2010feedback}) as well as neural networks \cite{recht2019tour}. 

We denote by $X\subset \mathbb{R}^n$ the state space of System \eqref{eq:system_equation}, which is assumed to be a compact set. However, since in general $\mathbf{x}[k]$ is not guaranteed to always
lie inside a compact set (e.g. if the noise distribution $p(v)$ has unbounded support), as common in the literature \cite{SANTOYO2021109439,prajna2007framework}, we consider the  stopped process $\tilde{\textbf{{x}}}[k]$ defined as follows. 
\begin{definition}{(Stopped Process)}
    Let $\tilde{k}$ be the first exit time of $\textbf{{x}}$ from $X$. Then, the stopped process $\tilde{\textbf{{x}}}[k]$ is defined as 
    $\tilde{\textbf{{x}}}[k]=\begin{cases}
\mathbf{x}[k] \quad \text{if }k<\tilde{k}\\
\mathbf{x}[\tilde{k}] \quad \text{otherwise}
\end{cases}.$
\end{definition}

For a given initial condition $x_0$, $\tilde{\textbf{{x}}}[k]$ is a Markov process with a well defined probability measure $P$ generated by the noise distribution $p(v)$ \cite[Proposition 7.45]{bertsekas2004stochastic} such that for sets $X_0, X_{k+1} \subseteq X$ it holds that
\begin{equation}
\begin{aligned}
    &P(\tilde{\textbf{{x}}}[0] \in X_0) = \mathbbm{1}_{X_0}(x_0)\\
    &P(\tilde{\textbf{{x}}}[k + 1] \in X_{k+1} \mid \tilde{\textbf{{x}}}[k] = x_k) = \int_V \mathbbm{1}_{X_{k+1}}(F(x_k) + v_k) \cdot p(v_k)\,dv_k,
\end{aligned}
\end{equation}
where $\mathbbm{1}_{X_k}(x_k)=\begin{cases}
1 \quad \text{if }x_k \in X_k\\
0 \quad \text{otherwise}
\end{cases}$ is the indicator function for set $X_k.$ 


In this paper we focus on verifying the safety of System \eqref{eq:system_equation} defined as the probability that for a given finite time horizon $H\in \mathbb{N}$,  $\tilde{\textbf{{x}}}[k]$ remains within a safe set $X_s\subseteq X$, which we assume to be a measurable set.  
\begin{problem}[Probabilistic Safety]
    Given a safe set $X_s\subseteq X$, a finite time horizon $H$, and an initial set of states $ X_0 \subseteq X_s$, compute
    \[
        P_{safe}(X_s,X_0,H)=\inf_{x_0\in X_0}P(\forall k\in [0,H], \tilde{{\mathbf{x}}}[k] \in X_s \mid x[0]=x_0)
    \]
\end{problem}
Note that the assumption of a finite time horizon is not limiting. In fact, if $p(v)$ has unbounded support, the probability of entering any unsafe region over an unbounded horizon is trivially $1$.
Furthermore, we should stess that the distribution of $\tilde{\textbf{{x}}}[k]$ is analytically intractable, because $\tilde{\textbf{{x}}}[k]$ is the result of iterative predictions over a non-linear function $F$ with additive noise, which is analytically intractable even if $p(v)$ is Gaussian \cite{girard2002gaussian}. Consequently, the computation of $ P_{safe}(X_s,X_0,H)$ is particularly challenging and requires approximations. Our approach is to rely on barrier functions parameterized as neural networks to compute a sound lower bound of $P_{safe}$. 

\section{Background on Local Relaxations of Neural Networks}\label{sec:prelim_crown}
\Glspl{nn} are highly non-linear and more importantly non-convex functions \cite{li2018visualizing}. Hence, in order to prove that a neural network satisfies the conditions to be a barrier function we rely on local upper and lower approximations of the \gls{nn}, also known as relaxations. 
The simplest type of relaxation is interval bounds where given a hyperrectangular input set, the bounds are an interval that contains all the outputs of the \gls{nn} for the points in the input set \cite{wicker2020probabilistic}.
\begin{definition}[Interval relaxation]\label{def:interval_relaxation}
    An interval relaxation of a continuous function $f : \mathbb{R}^n \to \mathbb{R}^m$ over a set $X\subseteq \mathbb{R}^n$ are two vectors\footnote{The two symbols $\bot,\top$ are called bottom and top respectively.} $b^{\bot}, b^{\top} \in \mathbb{R}^m$ such that
    $
        b^{\bot} \leq f(x) \leq b^{\top}, \quad \forall x \in X.
   $
\end{definition}
An alternative relaxation that often produce tighter bounds are linear relaxations.
\begin{definition}[Linear relaxation]\label{def:linear_relaxation}
    A linear relaxation of a continuous function $f: \mathbb{R}^n \to \mathbb{R}^m$ over a set $X\subseteq \mathbb{R}^n$ are two linear functions $A^{\bot}x + b^{\bot}$ and $A^{\top}x + b^{\top}$ with $A^{\bot}, A^{\top} \in \mathbb{R}^{m\times n}$ and $b^{\bot}, b^{\top} \in \mathbb{R}^m$ such that
    $
        A^{\bot}x + b^{\bot} \leq f(x) \leq A^{\top}x + b^{\top}, \quad \forall x \in X.
    $
\end{definition}
If a linear relaxation of a function $f$ is defined over a hyperrectangular set $X$, then it is possible to find an interval relaxation from the linear relaxation \cite{zhang2020towards}.
Let $X = \{x \in \mathbb{R}^n \mid x^\bot \leq x \leq x^\top\}$ be a hyperrectangle and $A^{\bot}x + b^{\bot} \leq f(x) \leq A^{\top}x + b^{\top}$ denote a linear relaxation. Then an interval relaxation $b^{\bot}_{interval}, b^{\top}_{interval}$ can be computed as
\begin{align}
    b^{\bot}_{interval} &= A^{\bot}\left(\frac{x^\top + x^\bot}{2}\right) - |A^{\bot}|\left(\frac{x^\top - X^\bot}{2}\right) + b^{\bot}\label{eq:linear2interval_lower}\\
    b^{\top}_{interval} &= A^{\top}\left(\frac{x^\top + x^\bot}{2}\right) + |A^{\top}|\left(\frac{x^\top - x^\bot}{2}\right) + b^{\top}\label{eq:linear2interval_upper}
\end{align}
In order to find a linear relaxation of a neural network with general activation functions (assumed to be continuous) we employ CROWN \cite{zhang2018efficient}, where linear lower and upper bounds are propagated backwards through the neural network architecture. 
We note that if a \gls{nn} is composed with a continuous function, then CROWN-like techniques can still be applied on the composite computation graph to derive linear relaxations of the composed function. In particular, we can treat the continuous function as the first layer of neural network and perform backwards bound propagation \cite{xu2020}. 




\section{Probabilistic Safety via Neural Barrier Functions}\label{sec:safe_nbf}

Our framework to compute probabilistic safety for System \eqref{eq:system_equation} is based on stochastic barrier functions, which we parameterize as neural networks. Since this paper only focuses on stochastic barrier functions, we sometimes refer  to them as just barrier functions. In what follows, we first introduce stochastic barrier functions (Section \ref{sec:dt_sbf}) and then show in Section \ref{sec:method_verify} how to verify that a neural network is a barrier function for System \eqref{eq:system_equation} using the relaxation methods introduced in Section \ref{sec:prelim_crown}.
Section \ref{sec:method_verify_bab} improves the verification by introducing a branch-and-bound partitioning method to find tighter bounds for $P_{safe}$, while keeping under control the required computational cost.
Finally, in Section \ref{sec:method_train} we show how techniques commonly used for robust training of NNs can be used to train neural barrier functions (NBFs) for System \eqref{eq:system_equation}. 

\subsection{Stochastic barrier functions}\label{sec:dt_sbf}


Similar to Lyapunov functions for proving stability \cite{steinhardt2012finite}, the main idea of stochastic barrier function is to study the time properties of a system without the need to compute its flow explicitly. In particular, stochastic barrier functions rely on the theory of  $c$-martingales\footnote{In the rest of the paper we will use $\beta$ instead of $c$ as is custom for stochastic barrier functions.} to show that a stochastic process does not exit a given safe set with high probability.

\begin{definition}[Stochastic Barrier Function]\label{def:sbf}
Let $X_s \subseteq X,$ $X_0\subseteq X_s$ and $X_u=X\setminus X_s$ be respectively safe set, set of initial states, and unsafe set. Then, we say that a non-negative continuous almost everywhere function $B:\mathbb{R}^n \to \mathbb{R}_{\geq 0}$ is a stochastic barrier function for a stochastic discrete-time system $\tilde{\mathbf{x}}[k]$ as defined in Section \ref{sec:ProbForm} if there exists $ \beta \geq 0$ and $\gamma \geq 0$ such that

\begin{subequations}
    \begin{minipage}{0.365\textwidth}
        \vspace{-2mm}
        \begin{align}
B(x) \geq 0 \qquad &\forall x\in X\label{eq:barrier_ss}\\
    B(x) \geq 1 \qquad &\forall x\in X_u\label{eq:barrier_unsafe}
        \end{align}
    \end{minipage}
    \qquad
    \begin{minipage}{0.58\textwidth}
        \vspace{-2mm}
        \begin{align}
            B(x) \leq \gamma \qquad &\forall x\in X_0\label{eq:barrier_initial}\\
    \mathbb{E}[B(F(x) +  \mathbf{v})] \leq  B(x) + \beta \qquad & \forall x\in X_s \label{eq:barrier_expectation}.
        \end{align}
    \end{minipage}
\end{subequations}
\end{definition}
\begin{proposition}{(\cite{steinhardt2012finite})}\label{prop:probability_of_failure}
Let $B$ be a barrier function for $\tilde{\mathbf{x}}[k]$ and $H\in \mathbb{N}$ be a time horizon. Then, for $\varepsilon = \gamma + \beta \cdot H$ it holds that $
        P_{safe}(X_s,X_0,H) \geq 1 - \varepsilon.$
\end{proposition}
\begin{figure}
    \centering
    \includegraphics[width=0.8\textwidth]{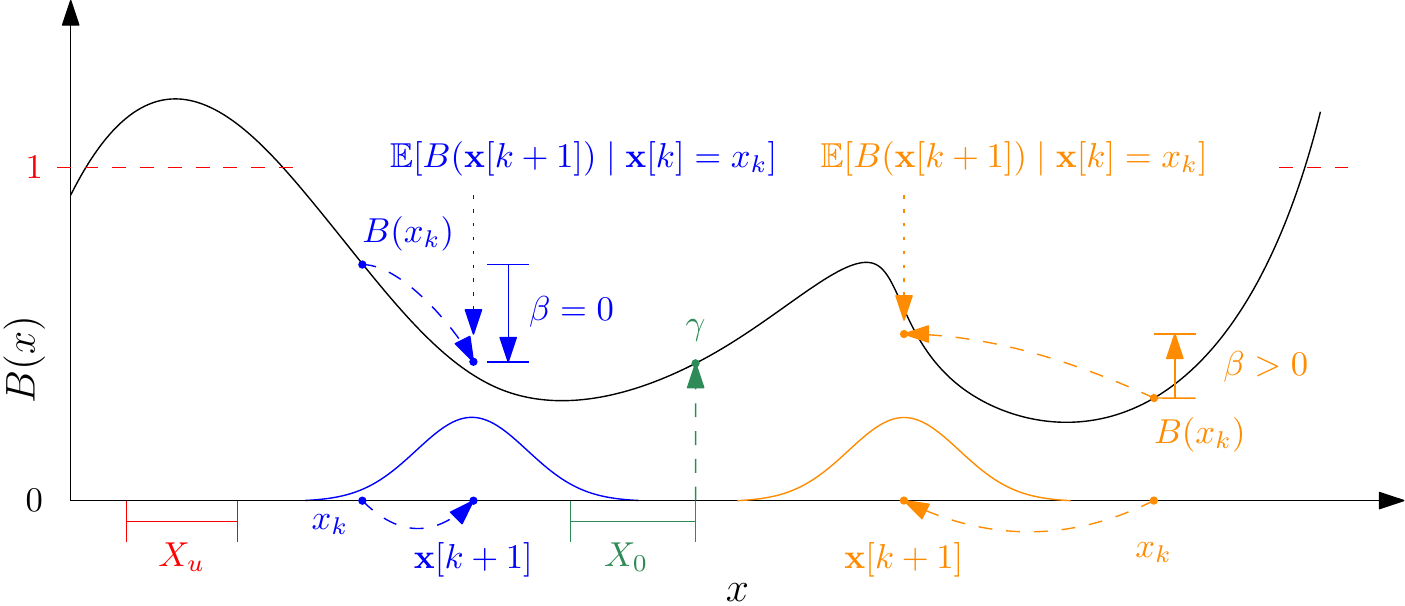}
    \caption{A  barrier function $B$ is a non-negative function that is greater than $1$ in the unsafe region. $\beta$ is an upper bound on the expected increase in value of the barrier function when composed with System \eqref{eq:system_equation} after one time step for any starting state $x_k$. If the the expected value is decreasing $\beta$ is taken to be zero. 
    $\gamma$ is an upper bound of $B(x)$ for $x\in X_0$. }
    \label{fig:example_barrier}
\end{figure}
An intuition behind Conditions \ref{eq:barrier_ss}-\ref{eq:barrier_expectation} is given in Fig. \ref{fig:polynomial_barrier}. Intuitively, these conditions guarantee that the expectation of the composition of $B$ with the one step dynamics of $\tilde{\mathbf{x}}$ does not grow by more than $\beta$ in $X_s$, i.e. it forms a $\beta-$martingale. This allows us to use non-negative martingale inequalities to compute  $P_{safe}(X_s,X_0,H)$ \cite{SANTOYO2021109439}. Critically, these are static conditions that do not require to evolve $\tilde{\mathbf{x}}[k]$ to study its behavior over time.

\subsection{Neural Stochastic Barrier Functions}\label{sec:method_verify}
Given a feed-forward \gls{nn} $B_\theta$ with arbitrarily many layers and  continuous activation functions
, where $\theta$ represents the vector of the parameters (weights and biases), we want to verify if $B_\theta$ is a valid stochastic barrier function for System \eqref{eq:system_equation}, i.e. if it satisfies Conditions \ref{eq:barrier_ss}-\ref{eq:barrier_expectation}, thus forming a neural barrier function (NBF). Our approach is based on employing the local relaxation techniques introduced in Section \ref{sec:prelim_crown} to build piece-wise linear functions that under- and over-approximate $B_\theta$. In particular, we partition $X$ to a finite set of regions $Q=\{q_1,...,q_{|Q|}\}$ and, using the techniques introduced in Section \ref{sec:prelim_crown}, for each $q\in Q$ we can find matrices $A^{\bot}_q, A^{\top}_q \in \mathbb{R}^{1\times n}$ and $b^{\bot}_q, b^{\top}_q \in \mathbb{R}$ such that
    \[
        A^{\bot}_qx + b^{\bot}_q \leq B_{\theta}(x) \leq A^{\top}_qx + b^{\top}_q, \quad \forall x \in q.
    \]
Then, the following lemma follows trivially.
\begin{lemma}
\label{lemma:Conditions}
Let $Q_{X_u}\subseteq Q$ and $Q_{X_0}\subseteq Q$ be such that $X_u \subseteq \cup_{q\in Q_{X_u}}q$ and $ X_0 \subseteq \cup_{q\in Q_{X_0}} q  $. Choose
$\gamma=\max_{q\in Q_{X_0}}\max_{x\in q} A^{\top}_qx + b^{\top}_q $. Then, if 
\begin{align}
   & \min_{q\in Q}\min_{x\in q} A^{\bot}_qx + b^{\bot}_q \geq 0 \qquad 
   & \min_{q\in Q_{X_u}}\min_{x\in q} A^{\bot}_qx + b^{\bot}_q \geq 1,
\end{align}
Conditions \ref{eq:barrier_ss}-\ref{eq:barrier_initial} are satisfied.
\end{lemma}
Note that under the assumption that $q$ is a convex polytope, which can always be enforced by the partition strategy, then Lemma \ref{lemma:Conditions} reduces to the solution of linear programs. 
Furthermore, as discussed in Section \ref{sec:prelim_crown}, we remark that if interval relaxation techniques are employed, then $\forall q \in Q$ $A^{\bot}_q=A^{\top}_q=0^{1\times n}$. As a consequence, checking Conditions \ref{eq:barrier_ss}-\ref{eq:barrier_initial} reduces to simply the evaluation of vectors $b^{\bot}_q,b^{\top}_q$ at the price of possibly more conservative bounds. 

We now turn our attention to $\beta$, and consequently to the computation of the martingale condition (Condition \ref{eq:barrier_expectation}).
Unfortunately, due to the non-linearity of the functions involved, for $x\in X$ $\mathbb{E}[B_\theta(F(x) + \mathbf{v}) ]$ is analytically intractable. As a consequence, we again rely on computing local under- and over-approximations. In particular, consider finite partitions $Q$ and $Q_V$  respectively of the state space $X$ and of the uncertainty space $V$, and let $\tilde{Q}=Q \times Q_v$. Then, as discussed in Section \ref{sec:prelim_crown} for each $q\in Q$ and $\tilde{q} = ({q}_x, {q}_v)\in \tilde{Q}$ we can find row vectors $A^{\bot}_{q},A^{\bot}_{{q}_x}, A^{\bot}_{{q}_v}, A^{\top}_{q}, A^{\top}_{{q}_x}, A^{\top}_{{q}_v} \in \mathbb{R}^{1\times n}$ and scalars $b^{\bot}_{{q}}, b^{\bot}_{\tilde{q}}, b^{\top}_{{q}}$, $ b^{\top}_{\tilde{q}} \in \mathbb{R}$ such that
\begin{align}
\label{eq:linearizationState}
    \forall x \in q,& \qquad  A^{\bot}_qx + b^{\bot}_q \leq B_{\theta}(x) \leq A^{\top}_qx + b^{\top}_q\\
    \forall (x',v') \in \tilde{q},& \qquad A^{\bot}_{{q}_x}x' + A^{\bot}_{{q}_v}v' + b^{\bot}_{\tilde{q}} \leq B_{\theta}(F(x')+v')  \leq  A^{\top}_{{q}_x}x'  + A^{\top}_{{q}_v}v'  + b^{\top}_{\tilde{q}}.
    \label{eq:linearizationNoise}
\end{align}
The following theorem uses the above relaxations to bound $\mathbb{E}[B(F(x) +  \mathbf{v})] $ and consequently find a lower bound on $\beta$.

\begin{theorem}
\label{th:main-Theorem}
    Let ${Q}$ and $Q_{V}$ respectively be partitions of $X$ and $V$. Let ${Q}_{X_s}\subseteq Q$ be such that $\cup_{q\in {Q}_{X_s}}q \subseteq X_s$. For $\tilde{q}=({q}_x,{q}_v)\in Q\times Q_{V}$ define 
    \begin{align*}
        A_{({q}_x,{q}_v)}=A^{\top}_{{q}_x} \int_{{q}_v}p(v)\,dv, \qquad
        b_{({q}_x,{q}_v)}=b^{\top}_{\tilde{q}} \int_{{q}_v}p(v)\,dv + A^{\top}_{{q}_v}\int_{{q}_v} v p(v)\,dv ,
    \end{align*}
 and assume
    \begin{align}
    \label{Eqn:Beta}
    \beta \geq \max_{q \in Q_{X_s}}  \max_{x \in q}\left( \big( \sum_{q_v \in Q_{V}} A_{(q,q_v)} - A^{\bot}_q \big) x +\big(\sum_{q_v \in Q_{V}}b_{(q,q_v)} - b^{\bot}_q\big) \right).
    \end{align}
    Then, for any  $x\in X_s$ it holds that $\mathbb{E}[B_\theta(F(x) + \mathbf{v}) ]-B_\theta(x) \leq \beta.$
\end{theorem}
The proof of Theorem \ref{th:main-Theorem} is reported in the Appendix  and relies on the under and over approximations introduced in Eqns \eqref{eq:linearizationState} and \eqref{eq:linearizationNoise}. In particular, by relying on the additive nature the noise and on the  linearity of the expectation, we can compute exactly how these linear functions are propagated through the expectation.

The computation of $A_{({q}_x,{q}_v)}$ and $b_{({q}_x,{q}_v)}$ in Theorem \ref{th:main-Theorem} requires the evaluation of integrals $\int_{q_v}p(v)\,dv$ and $\int_{q_v} v p(v)\,dv$, which are the probability of the noise being in $q_v$ 
and the partial expectation of the noise restricted to $q_v$ respectively. 
For various classes of distributions, such as Gaussian with diagonal covariance matrix, uniform, or finite support distributions, these integrals can be computed in closed forms. 
Otherwise, numerical approximations may be required.
One additional challenge is that if the noise has unbounded support as is the case with Gaussian noise, then some $q_v$ are infinite in size. 
For these partitions, linear relaxations may not exist, hence we cannot compute $A_{({q}_x,{q}_v)}$ and $b_{({q}_x,{q}_v)}$.
However, this problem can be solved by noticing that $\tilde{\mathbf{x}}$ is a stopped process outside $X$, which allows one to set $B_{\theta}(x)=0$ for all $x\not\in X$. With this assumption $B_{\theta}$ is still continuous almost everywhere (assuming that the boundary of $X$ has measure zero). This guarantees the conditions of Proposition \ref{prop:probability_of_failure} are satisfied.
Furthermore, since $X$ is bounded, such an assumption can simplify the partitioning as we can find the bounded subset of $V$ such that $B_{\theta}(F(x) + v)\neq 0$. 
In particular, for any $x\in X$, $B_\theta(F(x) + v)\neq 0$ only for $v\in V' = \{v \mid x^\bot - x^\top \leq v \leq x^\top - x^\bot\}$ where $x^\bot, x^\top \in \mathbb{R}^n$ are two vectors such that $X = \{x \mid x^\bot \leq x \leq x^\top\}$.

\begin{remark}
We remark that the results of this Section, and Theorem \ref{th:main-Theorem} in particular, can also be applied to systems with non-additive noise at the price of increased conservativeness. Specifically, given 
$    \mathbf{x}[k+1] = F(\mathbf{x}[k], \mathbf{v}[k])$ for $F:\mathbb{R}^n \times V \to \mathbb{R}^n$ continuous in both inputs, we can employ the linear relaxation techniques described in Section \ref{sec:prelim_crown} to find lower and upper bounds of the system dynamics that are linear in $x$ and $v$ locally to each partition. Then, the results in this Section can be applied.
\end{remark}
We stress that, similarly to Lemma \ref{lemma:Conditions}, Theorem \ref{th:main-Theorem} allows us to find $\beta$ by solving linear programs that reduces to evaluation of constants if interval relaxation techniques are employed. 

\subsubsection{A Branch and Bound Scheme for Verification}\label{sec:method_verify_bab}
\begin{figure}
    \centering
    \includegraphics[width=\linewidth]{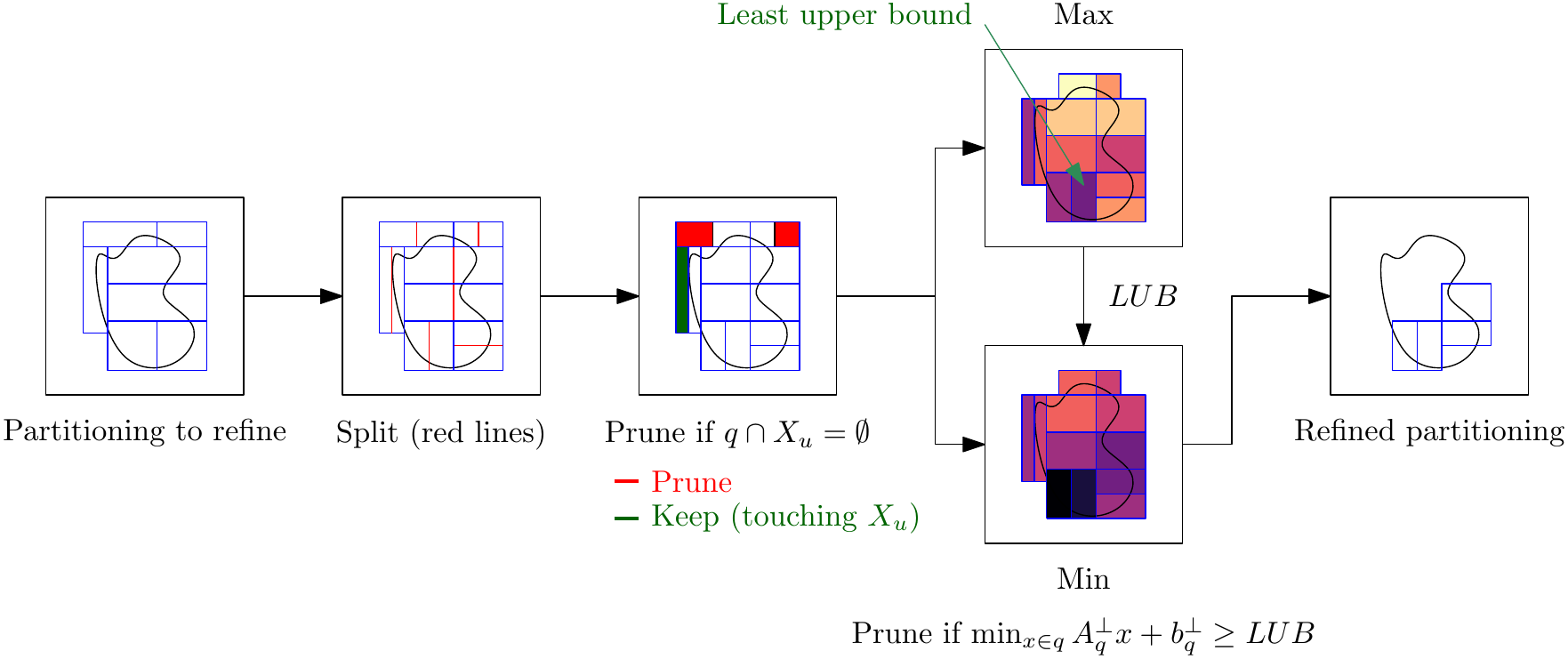}
    \caption{One iteration of the automatic partitioning scheme with splitting and pruning for Condition \ref{eq:barrier_unsafe}. The set $X_u$ is shown as a black blob, and hyperrectangles $q \in Q_{X_u}$ are split and pruned. 
    $LUB = \min_{q \in Q_{X_u}} \min_{x \in q} A^\top_q x + b^\top_q$ denotes the least upper bound.}
    \label{fig:automatic_partitioning}
\end{figure}
In order to guarantee scalability to our verification framework,  we develop a branch-and-bound partitioning scheme inspired by \cite{Bunel2018, xu2021fast} that starting from a coarse partitioning of $X$ gradually refines it by splitting regions and pruning those that already satisfy the barrier conditions.
 For convenience, we assume that all regions $q$ be hyperrectangles. 
 We perform the branch-and-bound independently for each of the conditions in Definition \ref{def:sbf} (Conditions \ref{eq:barrier_ss}-\ref{eq:barrier_expectation}).
 In what follows, we explain the partitioning scheme for Condition \ref{eq:barrier_unsafe}, the others follow analogously. 
 
 We start with a coarse initial partition $Q_{X_u}$ of $X_u$.  Then, as shown in Lemma \ref{lemma:Conditions}, for $X_u \subseteq \cup_{q\in Q_{X_u}} q $ Condition \ref{eq:barrier_unsafe} reduces to check if $\min_{q \in Q_{X_u}} \min_{x\in q} A^\bot_{q} x + b^\bot_{q} \geq 1.$
 As we start with a coarse partition initially our bounds may be very conservative.  Consequently, we gradually refine $Q_{X_u}$.
 First of all, we identify which regions to prune and which to split. This is decided based on the error introduced by the linear bounds in each partition. Specifically,  at each iteration we split all regions in $Q_{X_u}$, whereas we prune region $q$ if either $q \cap X_u = \emptyset$ or $\min_{x\in q} A^\bot_q x + b^\bot_q \geq \min_{q\in Q_{X_u}} \min_{x\in q} A^\top_q x + b^\top_q$. In fact, if the minimum value of $B$ in $q$ is greater than the smallest upper bound in another region $q' \in Q_{X_u}$, then $q$ does not influence the satisfaction of Condition \ref{eq:barrier_unsafe} and can be discarded.
 
 One iteration of splitting and pruning is shown in Fig. \ref{fig:automatic_partitioning}. For each region $q$ we only split in half the dimension $d$ that introduces the highest source of uncertainty, that is, we split dimension $d$ at the midpoint such that $
    d =\argmax_{1\leq i \leq n} \left((|A_{q}^{\bot}| + |A_{q}^{\top}|)^T \odot (q^{\top} - q^{\bot})\right)_i$
 where $\odot$ is the elementwise product and $q^{\bot}, q^{\top}$ denote the lower and upper bounds of $q$ and $(\cdot)_i$ represents the $i$\textsuperscript{th} component of a vector. Then, we prune regions that do not influence the final result of the minimization problem according to the conditions described above. 
  Finally, we stop the partitioning when the barrier condition is satisfied, i.e.,  $\min_{q\in Q_{X_u}}\min_{x\in q} A^\bot_q x + b^\bot_q \geq 1$, or the gap between upper and lower bound for $\min_{x\in X_u} B_\theta(x)$ is less than a threshold $t_{gap}>0$, that is if
  $   \min_{q\in Q_{X_u}} \min_{x\in q} A^\top_q x + b^\top_q - \min_{q\in Q_{X_u}} \min_{x\in q} A^\bot_q x + b^\bot_q < t_{gap}.$

\subsection{Training Stochastic Neural Barrier Functions}\label{sec:method_train}
We now describe the neural network training procedure, which is the key piece to obtain a valid stochastic barrier function $B_{\theta}.$ As Conditions \ref{eq:barrier_ss}-\ref{eq:barrier_expectation} needs to hold over regions in the state space, the rationale behind our approach is to adapt techniques commonly employed in certified adversarial training of NN \cite{zhang2020towards,wicker2021bayesian} to our setting.  Our training procedure starts by sampling independently $m$ training points from each set $X,X_0,X_s,$ and $X_u$.  We denote each of the resulting training sets by $Q_{X_s}^{(m)}$, $Q_{X_u}^{(m)}$, $Q_{X_0}^{(m)}$, $Q_{X}^{(m)}$. Furthermore, we independently sample $l$ vectors $v_1, \ldots, v_l$ from the noise distribution $p(v).$ Then, for training parameters $0\leq \kappa \leq 1$ and $\epsilon>0,$ and a time horizon $H\in \mathbb{N}$, the robust training loss $\mathcal{L}_{\epsilon}$  is defined as follows
 \begin{align*}
   \mathcal{L}_{\epsilon} &= (1- \kappa)\mathcal{L}_{violation} + \kappa(\gamma^{(m)} + \beta^{(m)} \cdot H)\\
 \nonumber  \mathcal{L}_{violation} &=\\
   \frac{1}{2m}&\left(\sum_{x \in Q_{X}^{(m)}} \max(0, -\min_{x':||x-x' ||_{\infty}\leq \epsilon}B_\theta(x')) + \sum_{x \in Q_{X_u}^{(m)}} \max(0, 1 - \min_{x':||x-x' ||_{\infty}\leq \epsilon}B_\theta(x'))\right)\\
\gamma^{(m)} &= \max_{x \in Q_{X_0}^{(m)}}\max_{x':||x-x' ||_{\infty}\leq \epsilon}B_\theta(x')\\   
\beta^{(m)} &= \max_{x \in Q_{X_0}^{(m)}}\max_{x':||x-x' ||_{\infty}\leq \epsilon} \big(\frac{1}{l}\sum_{j=1}^l B_\theta(F(x') + v_j) - B_\theta(x')\big).
\end{align*}
Intuitively, $\mathcal{L}_{violation}$ penalizes parameters $\theta$ that violates Conditions \ref{eq:barrier_ss} and \ref{eq:barrier_unsafe}, while minimizing $\gamma^{(m)} + \beta^{(m)} \cdot H$  maximizes the safety probability according to Proposition \ref{prop:probability_of_failure}. Consequently, $\kappa$ weights between having a valid stochastic barrier function and achieving tight probability bounds. 
$\min$ and $\max$ of $B_\theta$ over an $\epsilon-$ball around each training point are computed similarly to in Lemma \ref{lemma:Conditions} and Theorem \ref{th:main-Theorem} by employing the linear and interval relaxation techniques introduced in Section \ref{sec:prelim_crown}.
This also explains the role of $\epsilon$: small values of $\epsilon$ guarantee tighter approximations of $B_{\theta}$, while for larger values we obtain potentially looser bounds but that cover a larger portion of the state space.  

\section{Experimental Evaluation}\label{sec:experiments}
We evaluate our framework on three benchmarks: a 2-D linear system taken from \cite{SANTOYO2021109439}, the 2-D polynomial system shown in Figure \ref{fig:polynomial_system_barrier} from \cite{Abate2021}, and a 3-D discrete-time Dubin's car model \cite{Abate2021}, which is a non-polynomial system.  
In order to show the flexibility of our framework, for all systems we consider the same neural barrier function architecture: a feed-forward neural network with 3 hidden layers, 128 neurons per hidden layer, and ReLU activation functions. For computing linear relaxations, we use CROWN-IBP \cite{zhang2020towards} during training and CROWN \cite{zhang2018efficient} for verification. We employ a batch size $m = 250$ and train the neural network for 150 epochs with 400 iterations per epoch.
To gradually switch from maximizing probability of safety to prioritizing a valid barrier, we start with $\kappa = 1.0$ and exponentially decay with multiplicative factor of $0.97$ for each epoch.  
We implemented  our method in Python. 
For the \gls{sos} comparision, we have reimplemented the algorithm in \cite{SANTOYO2021109439} in Julia (1.7.2) with SumOfSquares.jl (0.5.0), and use Mosek (9.3.11).
Experiments are conducted on an Intel i7 6700K CPU with 16GB RAM and Nvidia GTX1060 GPU with 6GB VRAM. Further details can be found in the Supplementary Material including an analysis on the effect of $\epsilon$
\footnote{Code for both \gls{nbf} and \gls{sos} is available under GNU GPLv3 license at \url{https://github.com/DAI-Lab-HERALD/neural-barrier-functions}.}.

\paragraph{Certification Results}
\begin{wraptable}{r}{0.6\textwidth}
    \centering
    \caption{Certified lower bound for $P_{safe}$. Higher is better, and the best result for each system is highlighted in \textbf{bold}. \Gls{nbf} stands for Neural Barrier Function (our approach), while \gls{sos} is sum-of-square optimization. Cells with ''-'' denotes that \gls{sos} failed to compute a barrier.}
    \label{tab:certification_results}
    {\small \begin{tabular}{lrrr}
        \toprule
         & Linear & 2-D polynomial & Dubin's car \\
         Method & & & \\
         \midrule
         \gls{sos} (4) & 0.690906 & 0.000000 & - \\
         \gls{sos} (8) & 0.975079 & 0.232710 & - \\
         \gls{sos} (13) & 0.998405 & 0.681383 & - \\
         \gls{sos} (15) & 0.999761 & - & - \\
         NBF & \textbf{0.999969} & \textbf{0.991664} & \textbf{0.870272} \\
         \bottomrule
    \end{tabular}}
\end{wraptable}

To illustrate the efficacy of our framework, in Table \ref{tab:certification_results} we compare the lower bound of $P_{safe}$ obtained with our model with a sum-of-squares optimization based approach \cite{SANTOYO2021109439}, which arguably is the state-of-the-art for finding barrier functions for stochastic systems. For all the benchmarks we consider SoS polynomials of order up to 15. In all cases it is possible to observe that our approach based on neural barrier functions (NBF) outperforms SoS optimization in terms of the tightness of the bounds. For instance, in the Dubin's car model, arguably the hardest example we consider due to its non-polynomial nature, SoS either fails due to excessive memory requirements or return a trivial lower bound on $0$, while our framework obtains a lower bound of $0.87$. In contrast, in the linear system both SoS and our approach obtain a similar certified level of safety, but SoS is substantially faster (orders of minutes for the linear system), as our framework still requires to first train a neural network and then certify it (orders of few hours for all benchmarks as we used the same neural network architecture).
To understand the difference in certified safety, we study contour plots of the barrier function with the initial and unsafe sets for the 2-D polynomial (see Figure \ref{fig:contour}).
Note that the certified lower bound for $P_{safe}$ via Proposition \ref{prop:probability_of_failure} can be non-zero only in regions where $B(x)<1$.  
Interestingly, we observe this region is significantly smaller for \gls{sos} compared \gls{nbf}, which is attributed to the reduced expressivity of the SoS polynomial.
The differences in region sizes and the distances to the initial sets explain the certification result for the 2-D polynomial system we see in Table \ref{tab:certification_results}.

\begin{figure}
    \centering
    \begin{subfigure}[b]{0.40\textwidth}
        \centering
        \includegraphics[width=\textwidth]{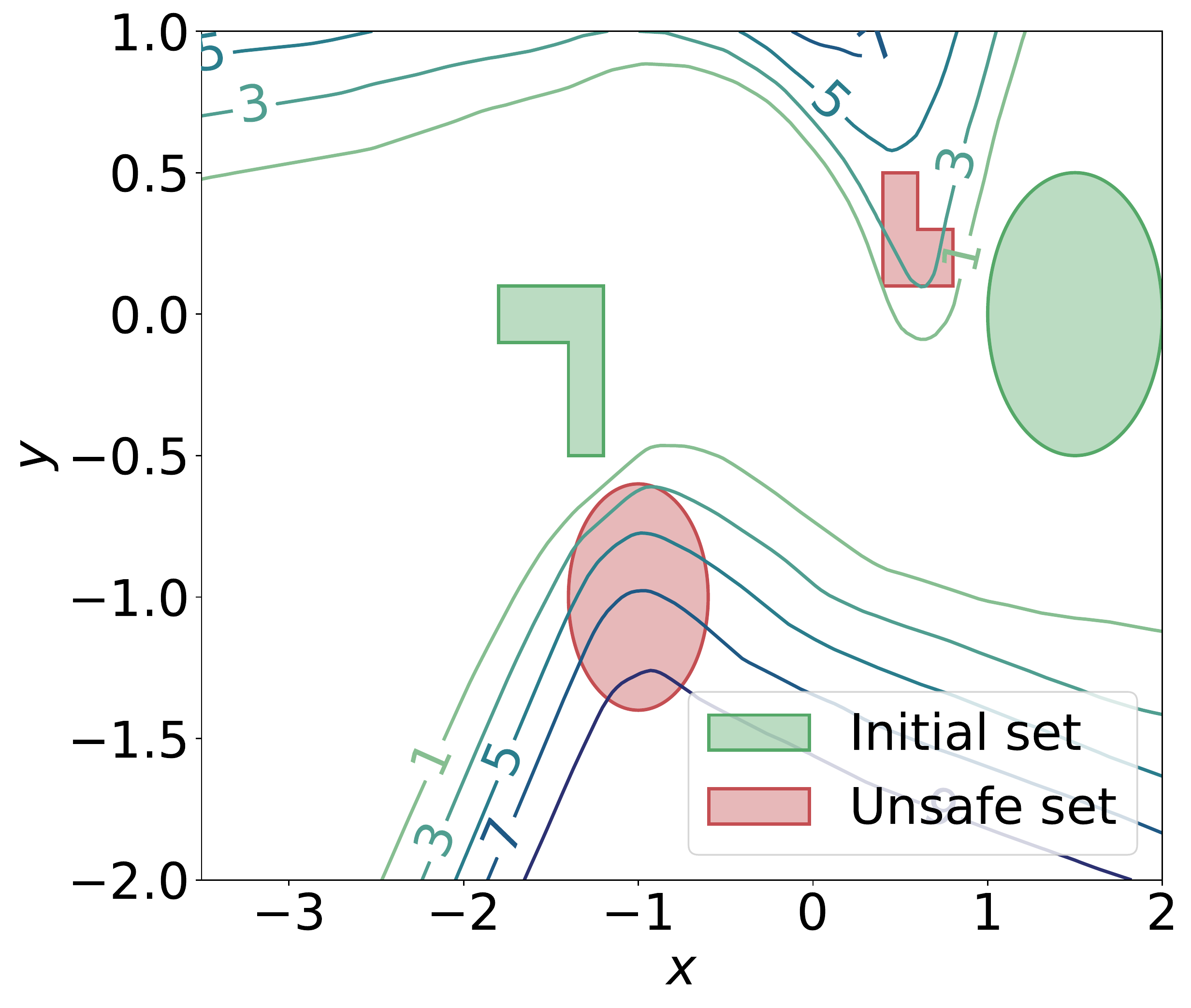}
        \caption{NBF}
    \end{subfigure}
    \begin{subfigure}[b]{0.40\textwidth}
        \centering
        \includegraphics[width=\textwidth]{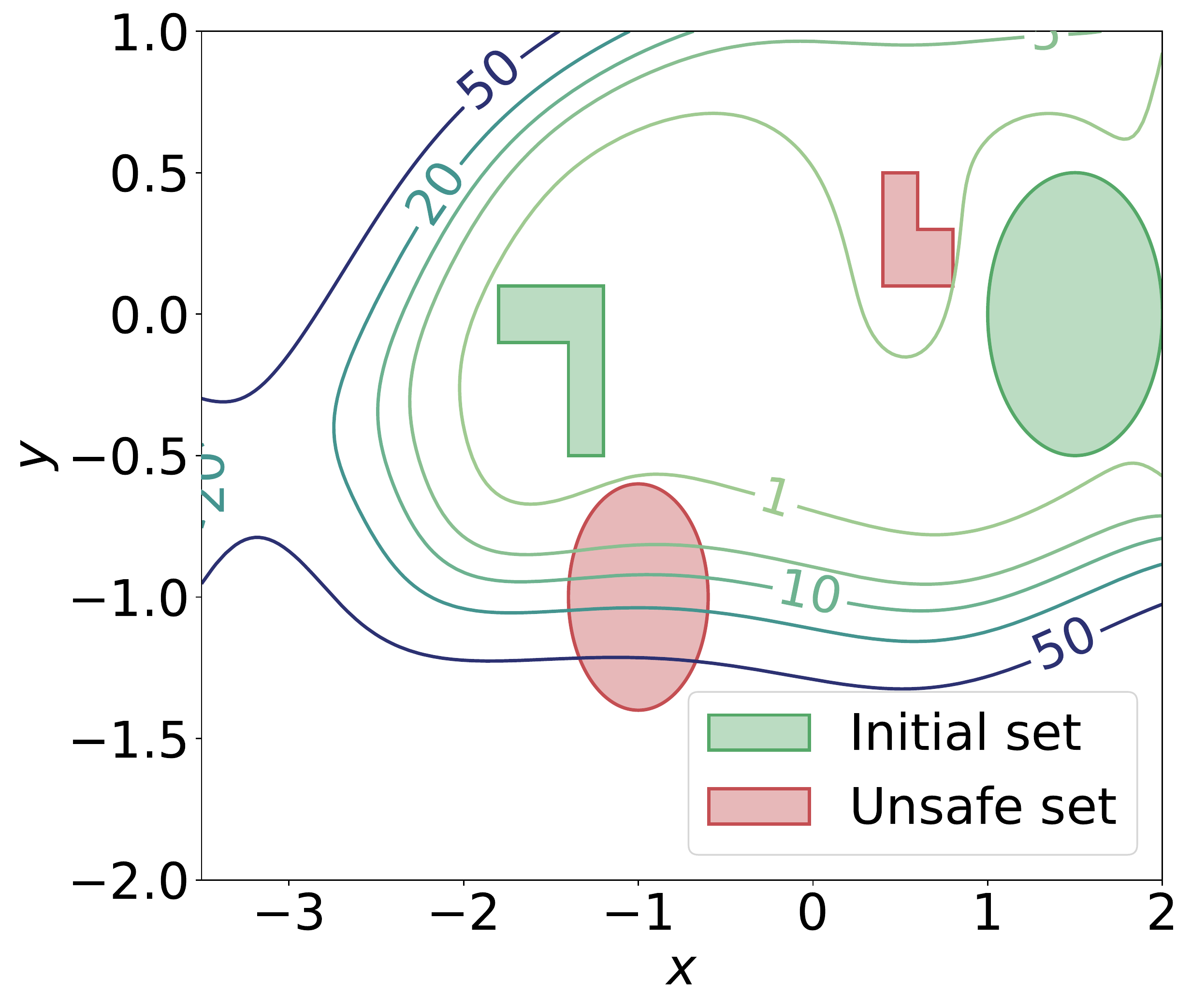}
        \caption{\gls{sos}}
    \end{subfigure}
    \caption{Levelset for $B_\theta$ \textbf{(a)} and $B_{SoS,8}$, an $8-$th order SoS barrier function \textbf{(b)} for the 2-D polynomial system. The neural network is more flexible to capture complex shapes of the unsafe set and less sharply increasing in value.}
    \label{fig:contour}
\end{figure}

\section{Conclusion}\label{sec:conclusion}
We studied probabilistic safety certification of stochastic systems using \acrfull{nbf}.
We presented algorithms to train NBFs and show that the problem of certifying that a neural network is a NBF for a given stocahstic system reduces to the solution of a set of linear programs. The scalability of our framework is guaranteed by a  branch-and-bound approach.
We evaluated our method on linear, polynomial, and non-linear and non-polynomial systems, beating state-of-the-art methods on all systems, thus certifying previously intractable non-linear systems. Hence, this works make a clear step towards the safe adoption of autonoumous systems in  safety-critical settings.
Future work may address scalability with probabilistic verification of barrier conditions and extend \glspl{nbf} to deterministic and continuous-time systems.

\printbibliography

\newpage

\appendix

\section{Appendix}

\subsection{Proof of Lemma \ref{lemma:Conditions}}
\setcounter{theorem}{1}
\begin{lemma}
Let $Q_{X_u}\subseteq Q$ and $Q_{X_0}\subseteq Q$ be such that $X_u \subseteq \cup_{q\in Q_{X_u}}q$ and $ X_0 \subseteq \cup_{q\in Q_{X_0}} q  $. Choose
$\gamma=\max_{q\in Q_{X_0}}\max_{x\in q} A^{\top}_qx + b^{\top}_q $. Then, if 
\begin{align}
   & \min_{q\in Q}\min_{x\in q} A^{\bot}_qx + b^{\bot}_q \geq 0 \qquad 
   & \min_{q\in Q_{X_u}}\min_{x\in q} A^{\bot}_qx + b^{\bot}_q \geq 1,
\end{align}
Conditions \ref{eq:barrier_ss}-\ref{eq:barrier_initial} are satisfied.
\end{lemma}
\begin{proof}
    
    \textbf{Condition \ref{eq:barrier_ss}.}
    Let $A_q^{\bot}x + b_q^{\bot}$ be a linear lower bound of $B_\theta(x)$ local to a region $q$, i.e. $A_q^{\bot}x + b_q^{\bot} \leq B_\theta(x)$ for all $x \in q$. 
    Then $\min_{x\in q}A_q^{\bot}x + b_q^{\bot}\leq \min_{x\in q}B_\theta(x)$, and therefore $\min_{q\in Q}\min_{x\in q}A_q^{\bot}x + b_q^{\bot}\leq \min_{q\in Q}\min_{x\in q}B_\theta(x)$ where $Q$ is a set of regions partitioning $X$.
    Since $Q$ is a partition of $X$, $\min_{q\in Q}\min_{x\in q}B_\theta(x) = \min_{x \in X} B_\theta(x)$. Therefore, we can conclude if $\min_{q\in Q}\min_{x\in q}A_q^{\bot}x + b_q^{\bot} \geq 0$ then $B_\theta(x) \geq 0$ for all $x \in X$ (Condition \ref{eq:barrier_ss}) is satisfied.
    
    \textbf{Condition \ref{eq:barrier_unsafe}.}
    Following the same proof structure, let $A_q^{\bot}x + b_q^{\bot}$ be a linear lower bound of $B_\theta(x)$ local to a region $q$.
    Then $\min_{q\in Q_{X_u}}\min_{x\in q}A_q^{\bot}x + b_q^{\bot}\leq \min_{q\in Q_{X_u}}\min_{x\in q}B_\theta(x)$ where $Q_{X_u} \subseteq Q$ is a set of regions covering $X_u$, i.e. $X_u \subseteq \cup_{q \in Q_{X_u}} q$.
    Since $Q_{X_u}$ is a partition covering $X_u$, minimum of $B_\theta$ of all regions in $Q_{X_u}$ is a lower bound for $B_\theta$ in $X_u$, \[\min_{q\in Q_{X_u}}\min_{x\in q}B_\theta(x) \leq \min_{x \in X_u} B_\theta(x).\]
    Therefore, if $\min_{q\in Q_{X_u}}\min_{x\in q}A_q^{\bot}x + b_q^{\bot} \geq 1$ then $B_\theta(x) \geq 1$ for all $x \in X_u$ (Condition \ref{eq:barrier_unsafe}) is satisfied.
    
    \textbf{Condition \ref{eq:barrier_initial}.}
    Once again following the same proof structure except for proving an upper bound, let $A_q^{\top}x + b_q^{\top}$ be a linear upper bound of $B_\theta(x)$ local to a region $q$.
    Then $\max_{q\in Q_{X_0}}\max_{x\in q}A_q^{\top}x + b_q^{\top}\geq \max_{q\in Q_{X_0}}\min_{x\in q}B_\theta(x)$ where $Q_{X_0} \subseteq Q$ such that $X_0 \subseteq \cup_{q \in Q_{X_0}} q$.
    Since $Q_{X_0}$ is a partition covering $X_0$,\[\max_{q\in Q_{X_0}}\max_{x\in q}B_\theta(x) \geq \max_{x \in X_0} B_\theta(x).\]
    Therefore, choosing $\gamma=\max_{q\in Q_{X_0}}\max_{x\in q} A^{\top}_qx + b^{\top}_q $ yields $\gamma \geq \max_{x\in X_0} B_\theta(x)$.
    We conclude $B_\theta(x) \leq \gamma$ for all $x \in X_0$ (Condition \ref{eq:barrier_initial}) is satisfied.
\end{proof}

\subsection{Proof of Theorem \ref{th:main-Theorem}}
We restate Theorem \ref{th:main-Theorem} and then prove it. Recall that for each $q\in Q$ and $\tilde{q} = ({q}_x, {q}_v)\in \tilde{Q}$ we let row vectors $A^{\bot}_{q},A^{\bot}_{{q}_x}, A^{\bot}_{{q}_v}, A^{\top}_{q}, A^{\top}_{{q}_x}, A^{\top}_{{q}_v} \in \mathbb{R}^{1\times n}$ and scalars $b^{\bot}_{{q}}, b^{\bot}_{(q_x,q_v)}, b^{\top}_{{q}}$, $ b^{\top}_{(q_x,q_v)} \in \mathbb{R}$ be such that
\begin{align*}
    \forall x \in q,& \qquad  A^{\bot}_qx + b^{\bot}_q \leq B_{\theta}(x) \leq A^{\top}_qx + b^{\top}_q\\
    \forall (x',v') \in \tilde{q},& \qquad A^{\bot}_{{q}_x}x' + A^{\bot}_{{q}_v}v' + b^{\bot}_{(q_x,q_v)} \leq B_{\theta}(F(x')+v')  \leq  A^{\top}_{{q}_x}x'  + A^{\top}_{{q}_v}v'  + b^{\top}_{(q_x,q_v)}.
\end{align*}
\begin{theorem}
    Let ${Q}$ and $Q_{V}$ respectively be partitions of $X$ and $V$. Let ${Q}_{X_s}\subseteq Q$ be such that $\cup_{q\in {Q}_{X_s}}q \subseteq X_s$. For $\tilde{q}=({q}_x,{q}_v)\in Q\times Q_{V}$ define 
    \begin{align*}
        A_{({q}_x,{q}_v)}=A^{\top}_{{q}_x} \int_{{q}_v}p(v)\,dv, \qquad
        b_{({q}_x,{q}_v)}=b^{\top}_{(q_x,q_v)} \int_{{q}_v}p(v)\,dv + A^{\top}_{{q}_v}\int_{{q}_v} v p(v)\,dv ,
    \end{align*}
 and assume
    \begin{align}
    \beta \geq \max_{q \in Q_{X_s}}  \max_{x \in q}\left( \big( \sum_{q_v \in Q_{V}} A_{(q,q_v)} - A^{\bot}_q \big) x +\big(\sum_{q_v \in Q_{V}}b_{(q,q_v)} - b^{\bot}_q\big) \right).
    \end{align}
    Then, for any  $x\in X_s$ it holds that $\mathbb{E}[B_\theta(F(x) + \mathbf{v}) ]-B_\theta(x) \leq \beta.$
\end{theorem}
\begin{proof}

By assumption it holds that for
each partition $\tilde{q}=({q_x},{q}_v)$
\begin{equation}
    B_{\theta}(F(x)+v) \leq A^{\top}_{{q_x}}x  + A^{\top}_{{q}_v}v  + b^{\top}_{(q_x,q_v)} \qquad \forall (x, v) \in \tilde{q}=({q_x},{q}_v)
\end{equation}
Hence, for $x\in q_x$ it holds that
\begin{equation}
    \begin{aligned}
        \mathbb{E}[B_\theta(F(x) + \mathbf{v}) \mid x] &=\sum_{q_v \in Q_{V}}\int_{q_v} B_\theta(F(x) + v)p(v)\,dv \\
                                                   &\leq \sum_{q_v \in Q_{V}}\int_{q_v} \left(A^{\top}_{{q_x}}x  + A^{\top}_{{q}_v}v  + b^{\top}_{(q_x,q_v)}\right)p(v)\,dv\\
                                                   & = \sum_{q_v \in Q_{V}} A_{({q_x},{q}_v)}x + b_{({q_x},{q}_v)}.
    \end{aligned}
\end{equation}
We can now combine the above bound with $A^{\bot}_qx + b^{\bot}_q$, the lower bound of $B_\theta(x)$. It then follows that
\begin{equation}
\begin{aligned}
    &\max_{x \in X_s}\big(\mathbb{E}[B_\theta(F(x) + \mathbf{v}) ]-B_\theta(x)\big) \leq \\
    &\qquad\qquad \max_{q \in Q_{X_s}}  \max_{x \in q}\left( \big( \sum_{q_v \in Q_{V}} A_{(q,q_v)} - A^{\bot}_q \big) x +\big(\sum_{q_v \in Q_{V}}b_{(q,q_v)} - b^{\bot}_q\big) \right)
\end{aligned}
\end{equation}

Therefore, if we pick $\beta$ such that
\begin{equation}
\max_{q \in Q_{X_s}}  \max_{x \in q}\left( \big( \sum_{q_v \in Q_{V}} A_{(q,q_v)} - A^{\bot}_q \big) x +\big(\sum_{q_v \in Q_{V}}b_{(q,q_v)} - b^{\bot}_q\big) \right) \leq \beta
\end{equation}
it holds that  $\mathbb{E}[B_\theta(F(x) + \mathbf{v}) ]-B_\theta(x) \leq \beta$.
\end{proof}

\subsection{Branch-and-bound algorithm}
Similarly to Section \ref{sec:method_verify_bab}, we only show the partitioning algorithm for verifying $B(x) \geq 1$ for all $x\in X_u$. The partitioning for the remaining barrier conditions (Condition \ref{eq:barrier_ss}, \ref{eq:barrier_initial}, \ref{eq:barrier_expectation}) follow analogously.
The algorithm starts from a coarse initial partition of $X_u$ called $Q_{init}$.
A possible initial partition is a single hyperrectangle encompassing $X$, which will always exist since $X$ is bounded.
Next, we find linear relaxations of $B_\theta$ for each region $q$ in the partition (Line 3-5), which we use for proving if $B_\theta(x) \geq 1$ for all $x \in X_u$.
This is the case if $\min_{q\in Q} \min_{x \in q} A_{q}^{\bot}x + b_{q}^{\bot} \geq 1$, meaning that we stop partitioning if this condition is satisfied (Line 7).
The other stop condition $\min_{q\in Q} \min_{x\in q} A^\top_q x + b^\top_q - \min_{q\in Q} \min_{x\in q} A^\bot_q x + b^\bot_q \leq t_{gap}$ is to stop the partitioning if the found lower bound is within $t_{gap}$ of the true minimum (Line 6).
If neither of the two stop conditions are satisfied, we refine the partition by splitting all regions (Line 8).
To select the split axis, we pick the one with the largest linear coefficients, because that maximizes tightening of both upper and lower bounds, weighted by the width of the region along the given axis to avoid elongated regions because that empirically yields loose bounds (Line 21-22)
Finally, to combat the exponential growth of splitting, we prune regions that cannot contain $\min_{x \in X_u} B(x)$, based on two conditions: $q \cap X_u = \emptyset$ meaning that a region has been split such that $q$ no longer overlaps with $X_u$, and the lower bound $\min_{x\in q} A_{q}^{\bot}x + b_{q}^{\bot}$ is larger than an upper bound for $\min_{x \in X_u} B(x)$ (Line 12-13).
\begin{algorithm}
\caption{Partitioning of unsafe set $X_u$ based on local linear relaxations of $B(x)$ to find $\min_{x \in X_u} B(x) $ given an initial partition $Q_{init}$ of $X_u$.}\label{alg:partitioning}
\begin{algorithmic}[1]
\Function{Partitioning-Unsafe}{$Q_{init}$, $t_{gap}$}
    \State $Q \gets Q_{init}$
    \For{Region $q$ in $Q$}
        \State $A_{q}^{\bot}, b_{q}^{\bot}, A_{q}^{\top}, b_{q}^{\top} \gets \Call{Crown}{B, q}$
    \EndFor
    
    \WhileNoDo{$\min_{q\in Q} \min_{x\in q} A^\top_q x + b^\top_q - \min_{q\in Q} \min_{x\in q} A^\bot_q x + b^\bot_q \geq t_{gap}$ \,\textit{and}}
        \State\hskip10pt $\min_{q\in Q} \min_{x \in q} A_{q}^{\bot}x + b_{q}^{\bot} \leq 1$ \algorithmicdo
        \State $Q \gets \Call{Split}{Q, f}$
        \For{Region $q$ in $Q$}
            \State $A_{q}^{\bot}, b_{q}^{\bot}, A_{q}^{\top}, b_{q}^{\top} \gets \Call{Crown}{B, q}$
        \EndFor
        \State $b_{lub} \gets \min_{q_i\in Q} \min_{x\in q} A_q^{\top}x + b_{q}^{\top}$ \Comment{Least upper bound}
        \State $Q \gets \{q \in Q \mid \min_{x\in q} A_{q}^{\bot}x + b_{q}^{\bot} \leq b_{lub} \text{ and } q \cap X_u \neq \emptyset \}$
    \EndWhile
    
    \State \textbf{return} $\min_{x \in q} A_{q}^{\bot}x + b_{q}^{\bot}$
\EndFunction
\Function{Split}{$Q$, $f$}
    \State $Q_{new} \gets \emptyset$
    \For{Region $q$ in $Q$}
        \State $A_{q}^{\bot}, b_{q}^{\bot}, A_{q}^{\top}, b_{q}^{\top} \gets \Call{Crown}{B, q}$
        \State $A_{q} \gets (|A_{q}^{\bot}| + |A_{q}^{\top}|) \odot (q^{\top} - q^{\bot})$ \Comment{Find axis with largest influence on bounds} 
        \State $d \gets \argmax_{1\leq i \leq n} A_{q, i}$ \Comment{Select axis to split} 
        \State $q_{1}, q_{2} \gets \Call{Split-Mid}{q, d}$ \Comment{Split hyperrectangle at midpoint along axis $d$}
        \State $Q_{new} \gets Q_{new} \cup \{q_{1}, q_{2}\}$
    \EndFor
    \State \textbf{return} $Q_{new}$
\EndFunction
\end{algorithmic}
\end{algorithm}

\subsection{Experiment details}\label{sec:experimental_details}
In this section, we describe the dynamics of each experimental system and the associated safe/unsafe and initial regions. 
For all systems, we define a step horizon $H = 10$ and the number for samples to approximate the expectation during training $l = 500$.


\paragraph{Linear dynamics}
We adopt the linear, discrete-time system from \cite{SANTOYO2021109439}, which represents juvenile/adult population dynamics \cite{Iannelli2014}. 
The system is governed following stochastic difference equation
\begin{equation}\label{eq:linear_dynamics}
\mathbf{x}[k+1] = \begin{bmatrix} 0 & m_3 \\ m_1 & m_2 \end{bmatrix} \mathbf{x}[k] + \mathbf{v}[k]
\end{equation}
We choose parameters $m_1 = 0.3$, $m_2 = 0.8$, and $m_3 = 0.4$, and $\mathbf{v}[k] \sim \mathcal{N}(\cdot \mid [0, 0]^T, [0, 0.1]^T)$, which we remark are different from \cite{SANTOYO2021109439}.
The parameters in \cite{SANTOYO2021109439} are unstable since the basic reproduction number $R = \frac{m_1 m_3}{1 - m_2}$ is $\frac{0.5\cdot 0.5}{1 - 0.95} = 5 > 1$, meaning the origin is an unstable equilibrium \cite{Iannelli2014}. 
We believe the results reported \cite{SANTOYO2021109439} are the result of a discrepancy in the associated code compared to the stochastic difference equation in the paper. More specifically, the columns in the top row of the dynamics matrix are swapped, resulting in one-way interaction rather than two-way interaction. 
We define the state space, initial and safe set as follows with $X_u = X \backslash X_s$
\begin{equation}
\begin{gathered}
    X = \{x \in \mathbb{R}^2 \mid -3 \leq x_1 \leq 3 \text{ and } -3 \leq x_2 \leq 3\} \\
    X_0 = \{x \in X \mid x_1^2 + x_2^2 \leq 1.5^2 \}, \qquad X_s = \{x \in X \mid x_1^2 + x_2^2 \leq 2^2\}
\end{gathered}
\end{equation}

\paragraph{Polynomial model}
For a polynomial system, we adapt 2-D system $\mathbf{barr}_3$ from \cite{Abate2021} by discretizing time using an Euler integrator and adding noise.
Due to the discretization, letting $h$ denote the step size, the time horizon is $H \cdot h$ with a step horizon $H$.
Particular for this system is that both the initial and unsafe sets consist of two disjoint regions, which is shown in Fig. \ref{fig:contour}.
A polynomial system can be directly encoded \gls{sos} optimization.

\begin{equation}
\begin{aligned}
\mathbf{x}[k+1]_{1}&= \mathbf{x}[k]_{1} + h \cdot \mathbf{x}[k]_{2} + \mathbf{v}[k]_{1}\\
\mathbf{x}[k+1]_{2}&= \mathbf{x}[k]_{2} + h \cdot \left(\frac{1}{3}\mathbf{x}[k]_{1}^3 - \mathbf{x}[k]_{1} - \mathbf{x}[k]_{2}\right) + \mathbf{v}[k]_{2} \\
\end{aligned}
\end{equation}
where $\mathbf{v}_k \sim \mathcal{N}(\cdot \mid [0, 0]^T, [0.01, 0]^T)$. 
The state space, initial and unsafe set (also shown in Figure \ref{fig:contour}) are defined as follows with $X_s = X \backslash X_u$
\begin{equation}
\begin{aligned}
    X &= \{x \in \mathbb{R}^2 \mid -3.5 \leq x_1 \leq 2 \text{ and } -2 \leq x_2 \leq 1\} \\
    X_0 &= circ(-1.5, 0, 0.5) \cup rect(-1.8, -0.1, 0.6, 0.2) \cup rect(-1.4, -0.5, 0.2, 0.6) \\
    X_u &= circ(-1, -1, 0.4) \cup rect(0.4, 0.1, 0.2, 0.4) \cup rect(0.4, 0.1, 0.4, 0.2)
\end{aligned}
\end{equation}
where $circ(a, b, r) = \{x \in \mathbb{R}^2 \mid (x_1 - a)^2 + (x_2 - b)^2 \leq r^2\}$ is a circle with radius $r$ centered at $(a, b)$ and $rect(a, b, c, d) = \{x \in \mathbb{R}^2 \mid a \leq x_1 \leq a + c \text{ and } b \leq x_2 \leq b + d\}$ is a rectangle with the lower corner at $(a, b)$ with width $c$ and height $d$.
We choose a step size $h = 0.1$.

\paragraph{Dubin's car}
System $\mathbf{barr}_4$ from \cite{Abate2021}, also known as Dubin's car \cite{8461017}, is our non-polynomial experimental system.
The state of system is the position in a plane and the heading of the vehicle, and the steering angle is bounded.
We adapt the system from continuous-time deterministic to discrete-time stochastic by discretizing time with an Euler integrator and adding noise to steering angle.
To encode the system into a polynomial suitable for \gls{sos}, we partition the state space into a grid, compute linear bounds of the nominal dynamics with CROWN, and use Putinar's Positivstellensatz \cite{Putinar1993}.

Let $h$ denote the step size, i.e. $H \cdot h$ is the time horizon. Then, dynamics are governed by
\begin{equation}
\begin{aligned}
\mathbf{x}[k+1]_{1}&= \mathbf{x}[k]_{1} + h \cdot v \sin(\mathbf{x}[k]_{3}) + \mathbf{v}[k]_{1}\\
\mathbf{x}[k+1]_{2}&= \mathbf{x}[k]_{2} + h \cdot v \cos(\mathbf{x}[k]_{3}) + \mathbf{v}[k]_{2} \\
\mathbf{x}[k+1]_{3}&= \mathbf{x}[k]_{3} + h \cdot u + \mathbf{v}[k]_{3} \\
\end{aligned}
\end{equation}
where $\mathbf{v}[k] \sim \mathcal{N}(\cdot \mid [0, 0, 0]^T, [0, 0, 0.01]^T)$, $v$ is the velocity, and $u$ denotes the steering angle. 
We choose a steering angle $u = 1 \,/\, 0.95$ such that the vehicle drives in a clockwise circle with a radius corresponding to distance between the origin and the starting position.
We additionally choose velocity $v=1$ and step size $h = 0.1$.
The state space, initial and unsafe set are defined as follows with $X_u = X \backslash X_s$
\begin{equation}
\begin{gathered}
    X = \{x \in \mathbb{R}^3 \mid -2 \leq x_1 \leq 2 \text{ and } -2 \leq x_2 \leq 2 \text{ and } -\pi / 2 \leq x_2 \leq \pi / 2 \} \\
    X_0 = \{(-0.95, 0, 0)\}, \qquad
    X_s = \{x \in X \mid -1.9 \leq x_1 \leq 1.9 \text{ and } -1.9 \leq x_2 \leq 1.9 \}
\end{gathered}
\end{equation}

\subsection{Study of \texorpdfstring{$\epsilon$}{}-hyperrectangles for training}
We analyze the impact of the varying $\epsilon$; half the width of the input hyperrectangle during training.
The analysis is conducted on the polynomial system as described in Sec. \ref{sec:experimental_details} because it has complex dynamics but is 2-D, so we may easily plot levelsets to study the impact.
Figure \ref{fig:epsilon_contour} shows contour plots of different $B_\theta(x)$ learned with various $\epsilon$ in increasing order.
The obvious change is that the increasing in barrier value over the state space is less for larger $\epsilon$, which is intuitive as the larger $\epsilon$ yields looser bounds hence the adversarial training promotes a flatter surface.
While a flat surface is good for a smaller $\beta$, it is a trade-off as it requires a larger $\gamma$ to ensure that the $B_\theta(x) \geq 1$ for $x \in X_u$.
Hence tuning $\epsilon$ is a trade-off between small $\beta$ and $\gamma$.

\begin{figure}
    \centering
    \begin{subfigure}[b]{0.4\textwidth}
        \centering
        \hspace*{-25pt}
        \includegraphics[width=\textwidth]{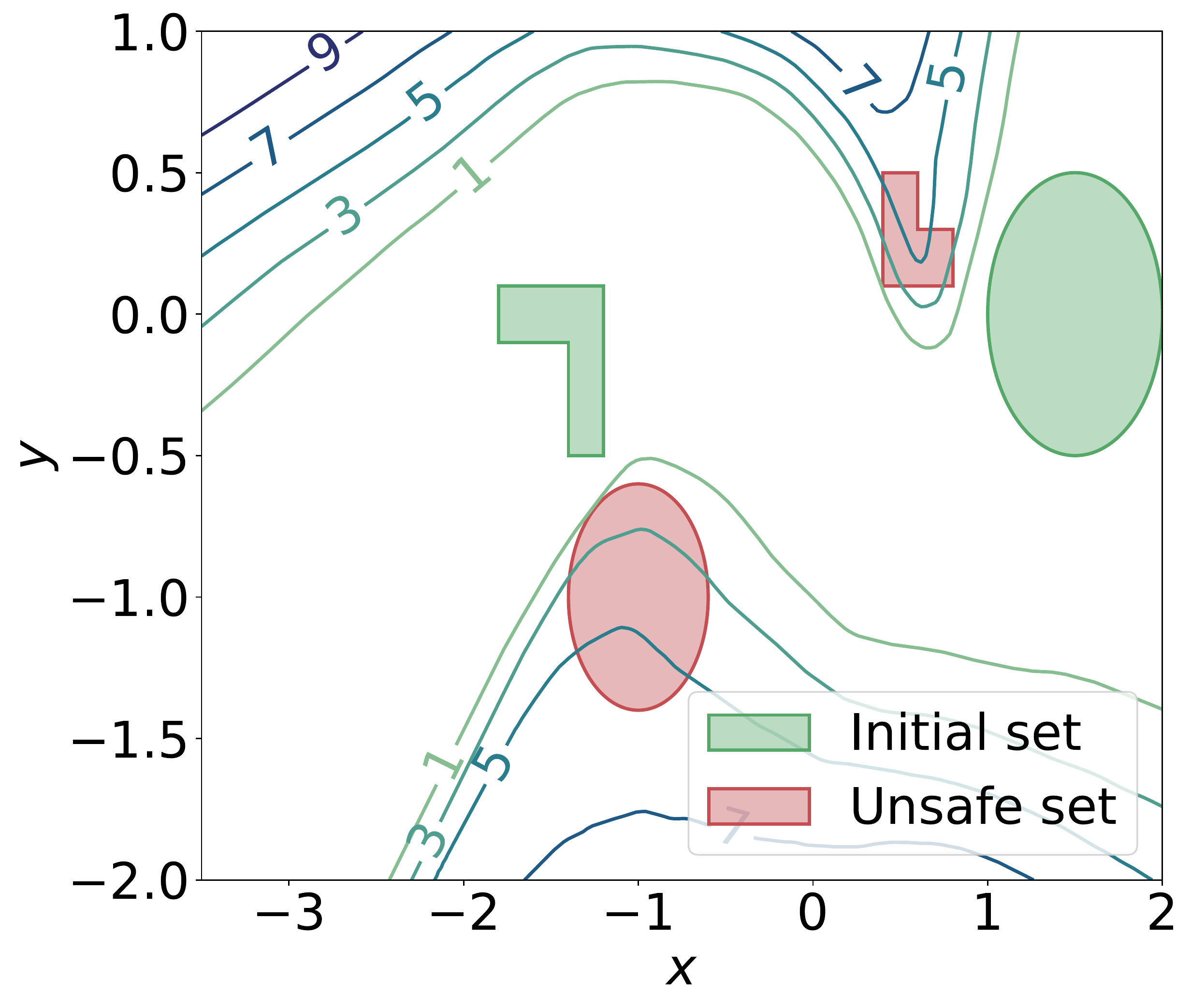}
        \caption{$\epsilon = 0.00001$}
    \end{subfigure}
    \begin{subfigure}[b]{0.4\textwidth}
        \centering
        \hspace*{-25pt}
        \includegraphics[width=\textwidth]{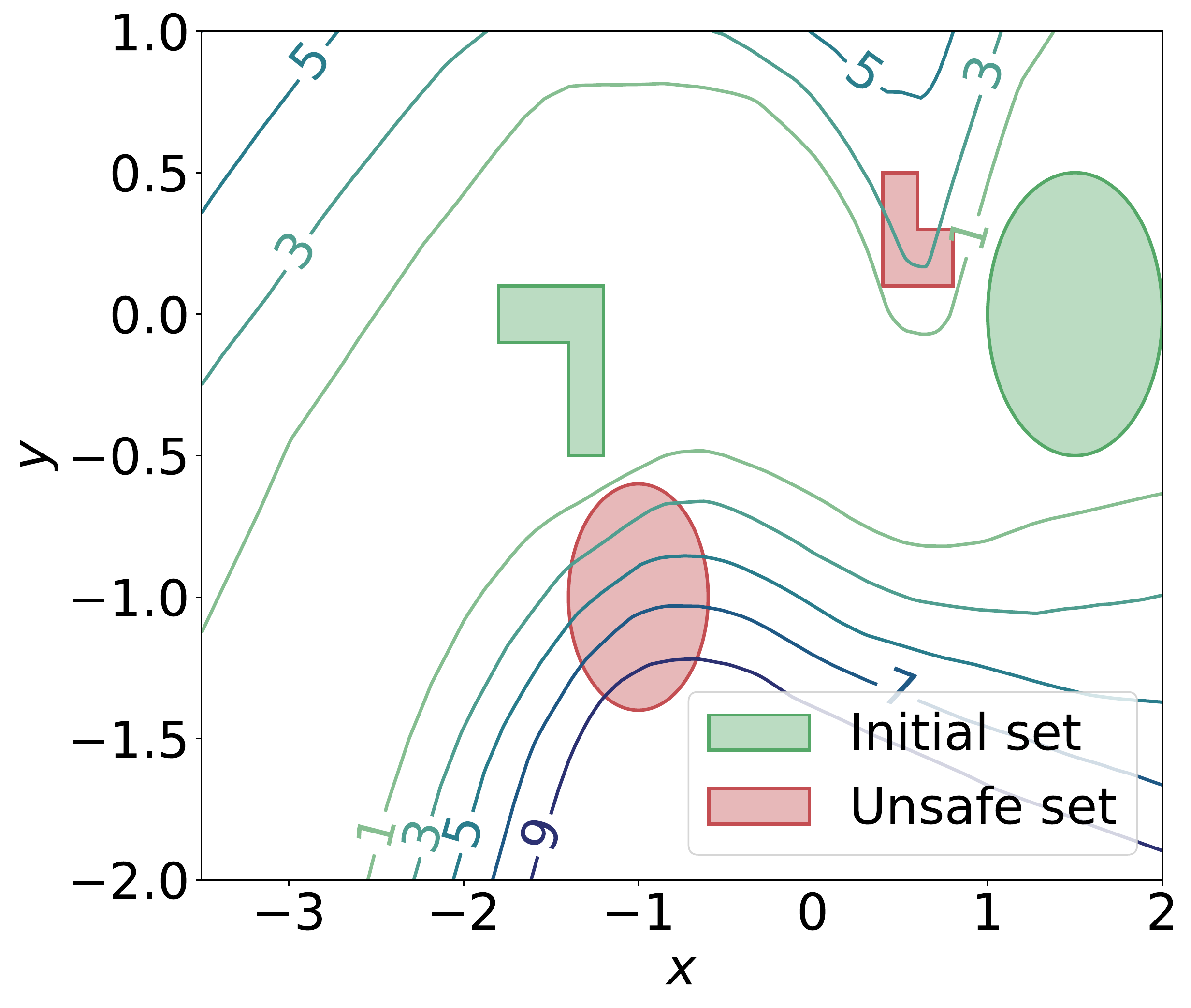}
        \caption{$\epsilon = 0.0001$}
    \end{subfigure}
    \begin{subfigure}[b]{0.4\textwidth}
        \centering
        \hspace*{-25pt}
        \includegraphics[width=\textwidth]{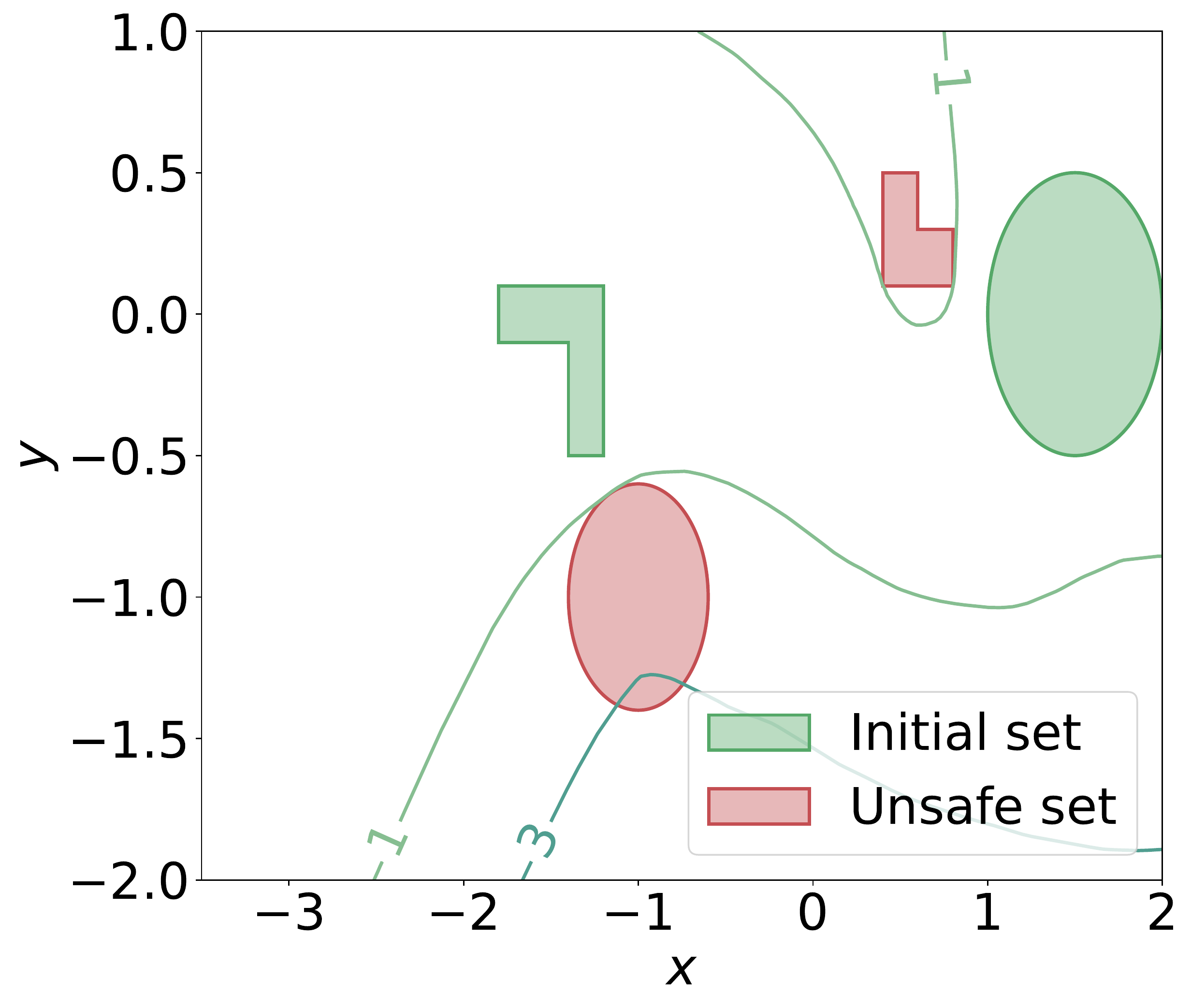}
        \caption{$\epsilon = 0.001$}
    \end{subfigure}
    \begin{subfigure}[b]{0.4\textwidth}
        \centering
        \hspace*{-25pt}
        \includegraphics[width=\textwidth]{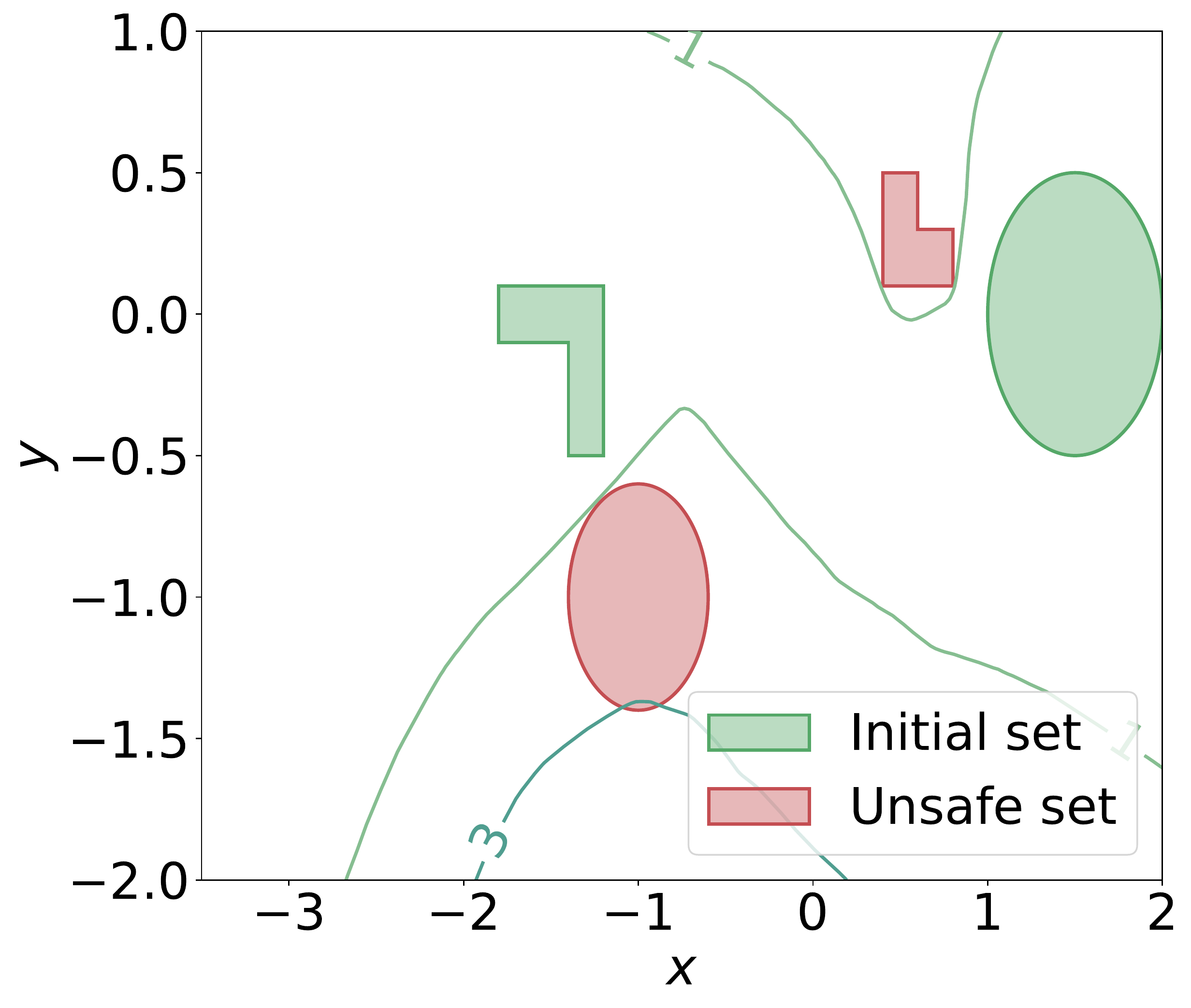}
        \caption{$\epsilon = 0.01$}
    \end{subfigure}
    \caption{Contour plots showcasing the impact of increasing $\epsilon$, the half width of the input hyperrectangle during training. Larger $\epsilon$ results in a flatter surface, yielding a smaller $\beta$ at the expense of a larger $\gamma$.}\label{fig:epsilon_contour}
\end{figure}

\end{document}

%% file: acronyms.tex
\newacronym{mdp}{MDP}{Markov Decision Process}
\newacronym{rss}{RSS}{Responsibility Sensitive Safety}
\newacronym{tpg}{TPG}{two-player game}
\newacronym{sos}{SoS}{Sum-of-Squares}
\newacronym{qcp}{QPC}{Quadratically Constrained Program}
\newacronym{sdp}{SDP}{Semi-definite Programming}
\newacronym{fcnn}{FCNN}{Fully-connected Neural Network}
\newacronym{bnn}{BNN}{Bayesian Neural Network}
\newacronym{mcmc}{MCMC}{Markov Chain Monte Carlo}
\newacronym{hmc}{HMC}{Hamiltonian Monte Carlo}
\newacronym{nuts}{NUTS}{No U-Turn Sampler}
\newacronym{dtss}{DT-SS}{Discrete-Time Stochastic System}
\newacronym{pac}{PAC}{Probably Approximately Correct}
\newacronym{psd}{PSD}{Positive Semi-Definite}
\newacronym{smt}{SMT}{Satisfiability Modulo Theory}
\newacronym{milp}{MILP}{Mixed-Integer Linear Programming}
\newacronym{lbp}{LBP}{Linear Bound Propagation}
\newacronym{ibp}{IBP}{Interval Bound Propagation}
\newacronym{nn}{NN}{Neural Network}
\newacronym{nbf}{NBF}{Neural Barrier Function}